\pdfoutput=1

\documentclass[USenglish, cleveref, a4paper]{lipics-v2019}
\nolinenumbers

\usepackage{graphicx}
\usepackage{mathtools}
\usepackage{extarrows}

\usepackage[utf8]{inputenc}
\usepackage{subcaption}
\usepackage[ruled, vlined , linesnumbered, noend]{algorithm2e}
\usepackage{amsthm}

\crefname{algocf}{Algorithm}{Algorithms}
\Crefname{algocf}{Algorithm}{Algorithms}

 \hideLIPIcs

\renewcommand{\emph}[1]{\textbf{\textit{#1}}}

\title{Reconstructing Biological and Digital Phylogenetic Trees in Parallel}
\titlerunning{Reconstructing Phylogenetic Trees in Parallel}

\author{Ramtin Afshar}{University of California-Irvine, USA}{afsharr@uci.edu}{https://orcid.org/0000-0003-4740-1234}{}
\author{Michael T.~Goodrich}{University of California-Irvine, USA}{goodrich@uci.edu}{https://orcid.org/0000-0002-8943-191X}{}
\author{Pedro Matias}{University of California-Irvine, USA}{pmatias@uci.edu}{https://orcid.org/0000-0003-0664-9145}{}
\author{Martha C.~Osegueda}{University of California-Irvine, USA}{mosegued@uci.edu}{https://orcid.org/0000-0002-1077-1074}{}
\authorrunning{R. Afshar, M. T. Goodrich, P. Matias, and M. C. Osegueda}
\Copyright{
Ramtin Afshar, Michael T.~Goodrich, Pedro Matias, and Martha C.~Osegueda}


\funding{This article reports on work supported by NSF grant 1815073.}
\EventEditors{Fabrizio Grandoni, Peter Sanders, and Grzegorz Herman}
\EventNoEds{3}
\EventLongTitle{28th Annual European Symposium on Algorithms (ESA 2020)}
\EventShortTitle{ESA 2020}
\EventAcronym{ESA}
\EventYear{2020}
\EventDate{September 7--9, 2020}
\EventLocation{Pisa, Italy (Virtual Conference)}
\EventLogo{}
\SeriesVolume{173}
\ArticleNo{70}
\begin{document}
\maketitle

\begin{abstract}

In this paper, we study the parallel query complexity of reconstructing 
biological and digital phylogenetic
trees from simple queries involving their nodes.
This is motivated from computational biology, 
data protection, and computer security 
settings, which can be abstracted
in terms of two parties, a \emph{responder}, Alice,
who must correctly answer queries of a given type
regarding a degree-$d$ tree, $T$,
and a \emph{querier}, Bob, who issues batches of queries, with each
query in a batch being independent of the others, so as to eventually
infer the structure of $T$.
We show that a querier can efficiently reconstruct an $n$-node
degree-$d$ tree, $T$, with a logarithmic number of rounds and 
quasilinear number of queries, with high probability, for various
types of queries, including \emph{relative-distance queries}
and \emph{path queries}.
Our results are all asymptotically optimal 
and improve the asymptotic (sequential) query complexity
for one of the problems we study.
Moreover, through an experimental analysis using both real-world
and synthetic data, we provide empirical evidence that our algorithms
provide significant parallel speedups while also
improving the total query complexities for the problems we study.

\end{abstract}

\setcounter{page}{1}

\clearpage

\section{Introduction}\label{sec:intro}
Phylogenetic trees
represent evolutionary relationships among a group of objects.
For instance,
each node in a biological phylogenetic tree represents a biological
entity, such as a species, bacteria, or virus, and the
branching represents how the entities are believed to have
evolved from common ancestors~\cite{DBLP:journals/jal/KannanLW96, DBLP:conf/icalp/BrodalFPO01, pagel1999inferring}.
(See \cref{fig:phylogen_life}.)
In a digital phylogenetic tree, 
on the other hand, each node represents a data object,
such as a computer virus~\cite{DBLP:journals/jal/GoldbergGPS98,DBLP:conf/malware/PfefferCCKOPZLGBHS12},
a source-code file~\cite{ji2008generating},
a text file or document~\cite{marmerola2016reconstruction,DBLP:journals/access/ShenFRS18},
or a multimedia object (such as an image or 
video)~\cite{DBLP:conf/icassp/BestaginiTT16,DBLP:journals/tifs/DiasRG12,DBLP:journals/ieeemm/DiasGR13,DBLP:journals/jvcir/DiasGR13} 
and the branching 
represents how these objects are believed to have evolved through edits
or data compression/corruption.
(See \cref{fig:phylogen_digital}.)
    
\begin{figure}[t!]
\begin{subfigure}[t]{0.49\textwidth}
  \centering
    \includegraphics[width=.85\linewidth]{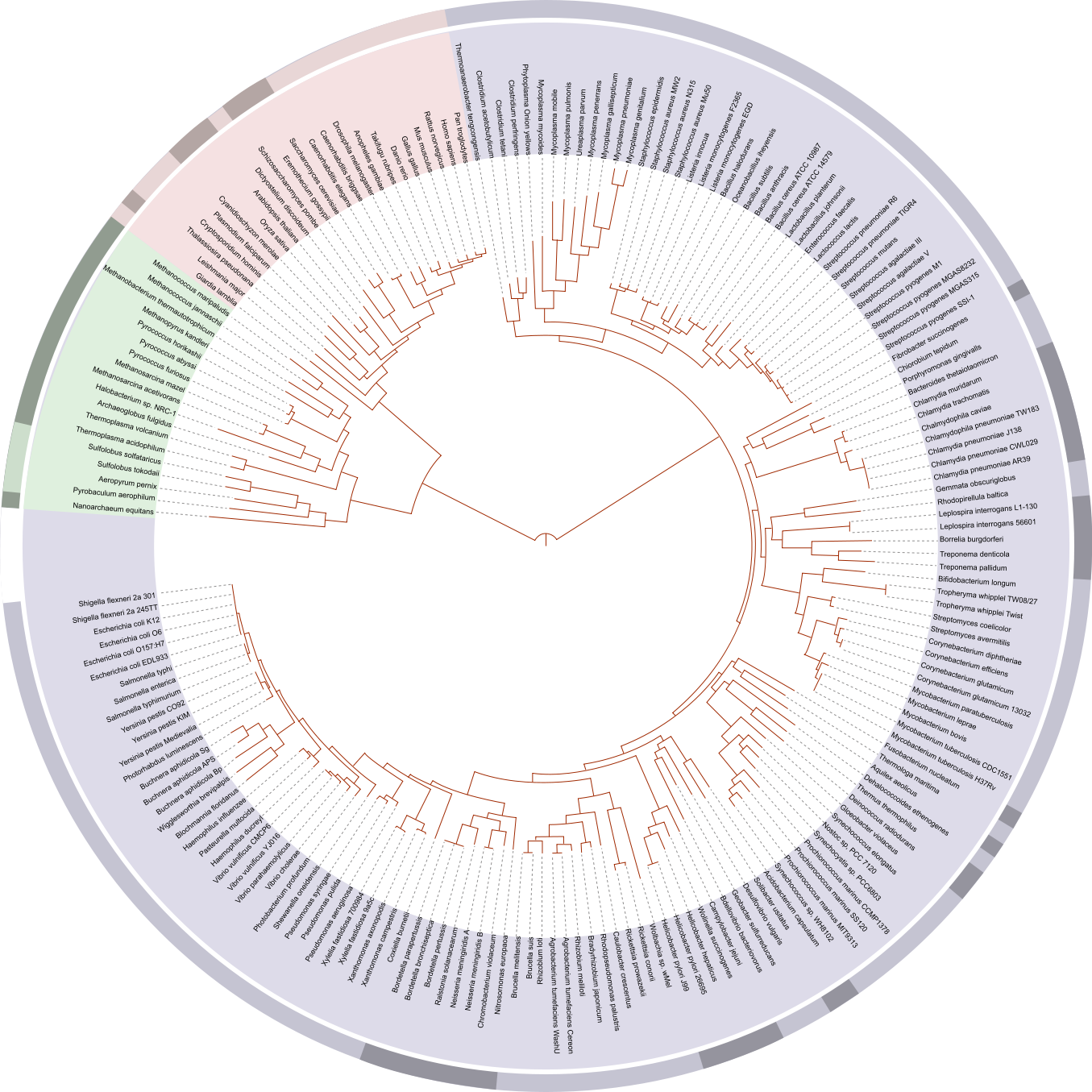}
    \caption{}
    \label{fig:phylogen_life}
\end{subfigure}
\begin{subfigure}[t]{0.49\textwidth}
      \centering
    \includegraphics[width=.85\linewidth]{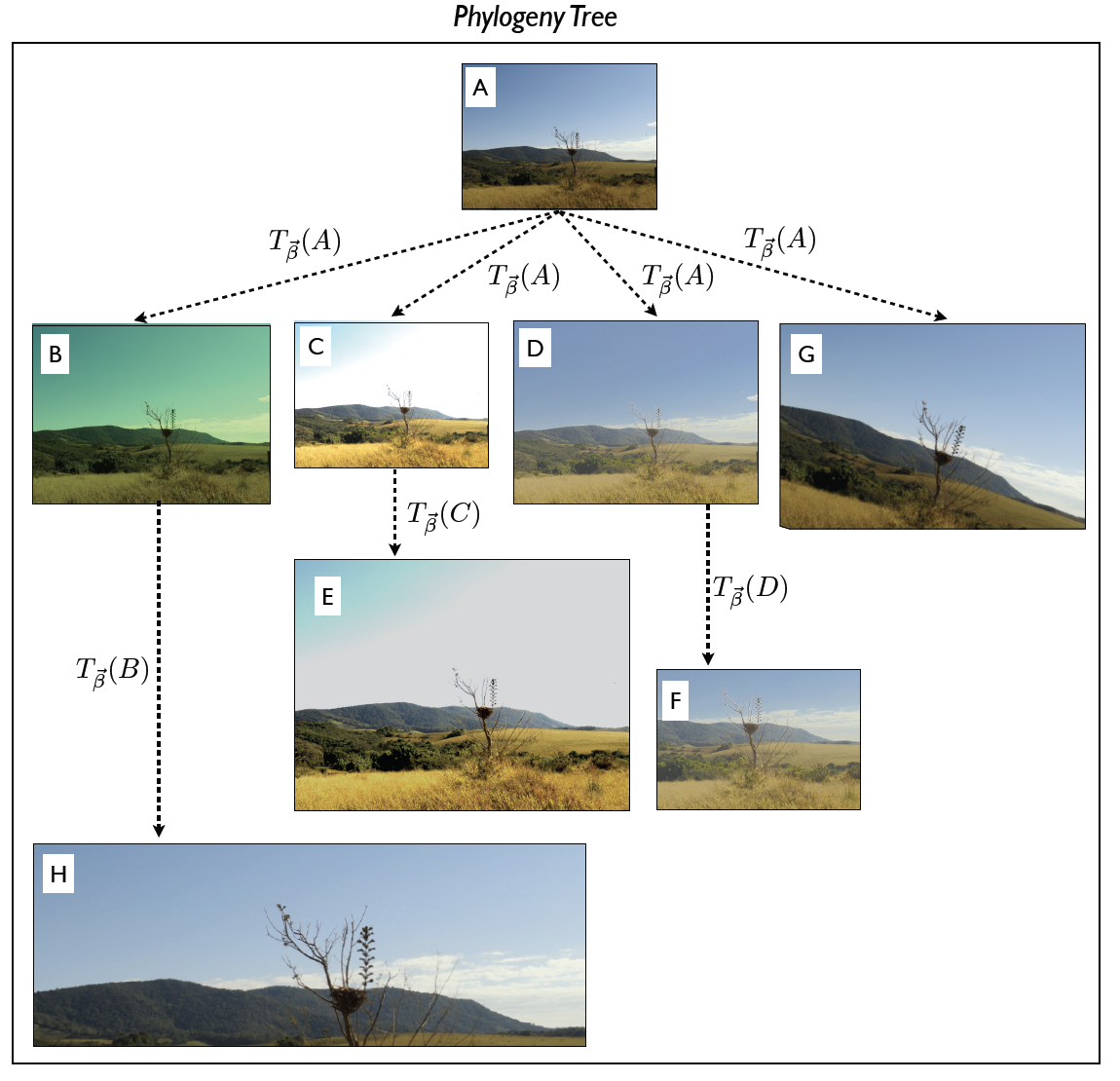}
    \caption{}
    \label{fig:phylogen_digital}
\end{subfigure}
\caption{Two phylogenetic trees.
(a) A biological phylogenetic tree of life, showing relationships between species
whose genomes had been sequenced as of 2006; public domain image
by Ivica Letunic, retraced by Mariana Ruiz Villarreal.
(b) A digital phylogenetic tree of images, from
Dias {\it et al.}~\protect{\cite{DBLP:journals/jvcir/DiasGR13}}.}
\label{fig:phylogen}
\end{figure}

In this paper, 
we are interested in studying efficient methods for 
reconstructing phylogenetic trees
from queries regarding their structure, noting that
there are differences in the types of queries one may
perform on the two types of phylogenetic trees.
In particular, with some exceptions,\footnote{One notable
  exception to this restriction of only being able to ask queries
  involving leaves in a biological phylogenetic tree is for phylogenetic
  trees of biological viruses, for which genetic sequencing may be known
  for all instances; hence,
  ancestor-descendant path queries might also be appropriate for
  reconstructing some biological phylogenetic trees.}
in a biological phylogenetic tree we can only perform queries 
involving the leaves
of the tree, since these typically represent living biological entities 
and internal nodes
represent ancestors that are likely to be extinct.
In digital phylogenetic trees, on the other hand, we can perform queries
involving any of the nodes in the tree, including internal nodes,
since these represent digital artifacts, which are often archived. The former type of phylogenetic tree has also received attention in the context of \textit{hierarchical clustering}, where the goal is to provide a hierarchical grouping structure of items according to their similarity \cite{DBLP:conf/soda/Emamjomeh-Zadeh18}.
To support reconstruction of both biological and digital phylogenetic trees,
therefore, we study both types of querying regimes in this paper.

More specifically, with respect to biological phylogenetic trees, we
focus on \emph{relative-distance} queries, 
where one is given three leaf
nodes (corresponding to species), $x$, $y$, and $z$, and the response is a 
determination of which pair, $(x,y)$, $(x,z)$, or $(y,z)$, is 
a closest pair, hence, has the most-recent common 
ancestor~\cite{DBLP:journals/jal/KannanLW96}.
With respect to digital phylogenetic trees, we instead focus on
\emph{path queries}, where one is given two nodes, $v$ and $w$, in the tree
and the response is ``true'' if and only if $v$ is an ancestor of $w$.

The motivation for reconstructing phylogenetic
trees comes from a desire to better understand
the evolution of the objects represented in a given phylogenetic tree.
For example, understanding how biological species evolved is useful for 
understanding and categorizing the fossil record and understanding when 
species are close 
relatives~\cite{DBLP:journals/jal/KannanLW96,DBLP:conf/soda/Emamjomeh-Zadeh18}.
Similarly, understanding how digital objects
have been edited and transformed can be useful for data protection,
computer security,
privacy, copyright disputes, and plariarism 
detection~\cite{DBLP:journals/jal/GoldbergGPS98,DBLP:conf/malware/PfefferCCKOPZLGBHS12,ji2008generating,
marmerola2016reconstruction,DBLP:journals/access/ShenFRS18,DBLP:conf/icassp/BestaginiTT16,DBLP:journals/tifs/DiasRG12,DBLP:journals/ieeemm/DiasGR13,DBLP:journals/jvcir/DiasGR13}.
For instance, understanding 
the evolutionary process of a computer virus
can provide insights into its ancestry, characteristics of the attacker, 
and where future attacks might come from and what they 
might look like~\cite{DBLP:conf/malware/PfefferCCKOPZLGBHS12}. 

The efficiency of a tree reconstruction algorithm can be characterized
in terms of its \emph{query-complexity} measure, $Q(n)$,
which is the total number of queries of a certain type needed to reconstruct
a given tree.
This parameter
comes from machine-learning and complexity theory,
e.g., see~\cite{DBLP:conf/birthday/AfshaniADDLM13,DBLP:journals/ai/ChoiK10,DBLP:conf/stoc/DobzinskiV12,DBLP:journals/combinatorica/Tardos89}, where it is also known as
``decision-tree complexity,'' e.g.,
see~\cite{DBLP:conf/stoc/Yao94,DBLP:journals/iandc/BernasconiDS01}.
Previous work on tree reconstruction
has focused on sequential methods,
where queries are issued and answered one at a time.
For example, in pioneering work for this research area, 
Kannan {\it et al.}~\cite{DBLP:journals/jal/KannanLW96} show that 
an $n$-node biological phylogenetic tree
can be reconstructed sequentially from $O(n\log n)$ 
three-node {relative-distance} queries.

Indeed, their reconstruction algorithms are inherently sequential and involve 
incrementally inserting leaf nodes into the phylogenetic tree reconstructed for
the previously-inserted nodes.

In many tree reconstruction applications, queries 
are expensive~\cite{DBLP:journals/jal/KannanLW96,DBLP:conf/soda/Emamjomeh-Zadeh18,%
DBLP:journals/jal/GoldbergGPS98,DBLP:conf/malware/PfefferCCKOPZLGBHS12,ji2008generating,%
marmerola2016reconstruction,DBLP:journals/access/ShenFRS18,DBLP:conf/icassp/BestaginiTT16,DBLP:journals/tifs/DiasRG12,DBLP:journals/ieeemm/DiasGR13,DBLP:journals/jvcir/DiasGR13},
but can be issued in batches.
For example, there is nothing preventing the 
biological experiments~\cite{DBLP:journals/jal/KannanLW96} 
that are represented in three-node relative-distance queries 
from being issued in parallel.
Thus, in order to speed up tree reconstruction,
in this paper we are interested in parallel tree reconstruction. 
To this end,
we also use a \emph{round-complexity}
parameter, $R(n)$, which measures the number of rounds of queries
needed to reconstruct a tree such that the queries issued in any round
comprise a batch of independent queries. That is, no query issued in a given round
can depend on the outcome of another query issued in that round, although both can
depend on answers to queries issued in previous rounds.
Roughly speaking, $R(n)$ corresponds to the span of a parallel reconstruction
algorithm and $Q(n)$ corresponds to its work. In this paper, we are interested in studying complexities for 
$R(n)$ and $Q(n)$ with respect to biological and digital phylogenetic 
trees with fixed maximum degree, $d$.

\subsection{Related Work}

The general problem of reconstructing graphs from distance queries was 
studied by Kannan {\it et al.}~\cite{DBLP:journals/talg/KannanMZ18},
who provide a randomized algorithm for reconstructing a graph of $n$
vertices using $\tilde O(n^{3/2})$
distance queries.\footnote{The $\tilde{O}(\cdot)$ notation hides poly-logarithmic factors.}

Previous parallel work has focused on inferring phylogenetic trees through Bayesian estimation \cite{DBLP:journals/bioinformatics/AltekarDHR04}.
However, we are not aware of previous parallel work using a similar query models to ours.
With respect to previous work on sequential tree reconstruction,
Culberson and Rudnicki \cite{DBLP:journals/ipl/CulbersonR89} provide the 
first sub-quadratic algorithms for reconstructing a weighted undirected tree 
with $n$ vertices and bounded degree $d$ from additive queries, 
where each query returns the sum of the weights of the edges of the path 
between a given pair of vertices. 
Reyzin and Srivastava~\cite{DBLP:journals/ipl/ReyzinS07} show that 
the Culberson-Rudnicki algorithm uses 
$O(n^{3/2} \cdot \sqrt{d})$ queries. 

Waterman {\it et al.}~\cite{waterman1977additive}
introduce the problem of reconstructing biological phylogenetic trees, 
using additive queries, which are more powerful than
relative-distance queries. Hein~\cite{hein1989optimal} 
shows that this problem has a solution 
that uses $O(d n\log_d n)$ additive queries, when the tree has maximum degree $d$, which is asymptotically optimal~\cite{DBLP:conf/soda/KingZZ03}.
Kannan {\it et al.}~\cite{DBLP:journals/jal/KannanLW96} show that 
an $n$-node binary phylogenetic tree can be reconstructed from $O(n\log n)$
three-node relative-distance queries.
Their method appears inherently sequential, however, as it is based on an
incremental approach that mimics insertion-sort.
Similarly, Emamjomeh-Zadeh and Kempe~\cite{DBLP:conf/soda/Emamjomeh-Zadeh18} also give
a sequential method using relative-distance queries 
that has a query complexity of $O(n\log n)$. Their algorithm, however, was designed for a different context, namely, hierarchical clustering.

Additionally, there exists some work (e.g. \cite{jones2004introduction,DBLP:conf/ccece/BhattacharjeeSS06a,huelsenbeck1995performance}) in an alternative perspective of the problem reconstructing phylogenetic trees, in which the goal is to find the best tree explaining the similarity and the relationship between a given fixed (or dynamic) set of data sequences (e.g. of species), using Maximum Parsimony \cite{10.1093/sysbio/19.1.83,10.2307/2412116,DBLP:journals/classification/Rohlf05} or Maximum Likelihood \cite{felsenstein1981evolutionary,DBLP:conf/ismb/ChorT05}. This contrasts with our approach of recovering the ``ground truth'' tree known only to an oracle, which is consistent with its answers about the tree.

With respect to digital phylogenetic tree reconstruction, there are a number
of sequential algorithms with $O(n^2)$ query complexities, including
the use of what we are calling path queries, where
the queries are also individually expensive, e.g., 
see~\cite{DBLP:journals/jal/GoldbergGPS98,DBLP:conf/malware/PfefferCCKOPZLGBHS12, ji2008generating,%
marmerola2016reconstruction,DBLP:journals/access/ShenFRS18,DBLP:conf/icassp/BestaginiTT16,DBLP:journals/tifs/DiasRG12,DBLP:journals/ieeemm/DiasGR13,DBLP:journals/jvcir/DiasGR13}. Jagadish and Sen~\cite{DBLP:conf/alt/JagadishS13}
consider reconstructing undirected 
unweighted degree-$d$ trees, giving
a deterministic algorithm that requires $O(dn^{1.5} \log n)$ separator queries, 
which answer if a vertex lies on the path between two vertices. 
They also give a randomized algorithm using an expected $O(d^2 n\log^2 n)$ number of
separator queries, 
and they give an $\Omega(dn)$ lower bound for any deterministic algorithm.
Wang and Honorio~\cite{DBLP:conf/allerton/WangH19}
consider the problem of reconstructing bounded-degree rooted trees, 
giving
a randomized algorithm that uses expected $O(dn\log^2 n)$ path queries.
They also prove that any randomized algorithm requires 
$\Omega(n\log n)$ path queries.

\medskip\noindent
\textsf{\textbf{Our Contributions.}}~~In this paper, 
we study the parallel phylogenetic tree reconstruction problem 
with respect to the two different types of queries
mentioned above:
\begin{itemize}
\item
We show that an $n$-node rooted biological (binary)
phylogenetic tree can be reconstructed 
from three-node \emph{relative-distance queries} with
$R(n)$ that is $O(\log n)$ and $Q(n)$ that is $O(n\log n)$, with high probability
(w.h.p.)\footnote{
  We say that an event occurs with high probability if it occurs with probability
  at least $1-1/n^c$, for some constant $c\ge 1$.}.
Both bounds are asymptotically optimal.

\item
We show that an $n$-node fixed-degree digital phylogenetic tree
can be reconstructed from \emph{path queries},
which ask whether a given node, $u$, is an ancestor of a given node, $w$,
with
$R(n)$ that is $O(\log n)$ and $Q(n)$ that is $O(n\log n)$, w.h.p. We also provide an $\Omega( d n + n \log n)$ lower bound for any randomized or deterministic algorithm suggesting that our algorithm is optimal in terms of query complexity and round complexity.
Further, this asymptotically-optimal $Q(n)$ bound
actually improves the sequential complexity for this problem,
as the previous best bound for $Q(n)$, due to
Wang and Honorio~\cite{DBLP:conf/allerton/WangH19}, had a $Q(n)$ bound
of $O(n\log^2 n)$ for reconstructing fixed-degree rooted trees using
path queries.
Of course, our method also
applies to biological phylogenetic trees that support path queries.
\end{itemize}

A preliminary announcement of some of this paper's results,
restricted to binary trees, was presented in \cite{DBLP:conf/spaa/AfsharGMO20}. Most of our algorithms are quite simple, although their analyses are at 
times nontrivial.
Moreover, given the many applications 
of biological and digital phylogenetic tree reconstruction, 
we feel that our algorithms have real-world applications.
Thus,
we have done an extensive experimental analysis of our algorithms,
using both real-world and synthetic data for biological and digital
phylogenetic trees. Our experimental results provide empirical
evidence that our methods achieve significant parallel speedups while
also providing improved query complexities in practice.

\section{Preliminaries}
In graph theory,
an \emph{arborescence} is a directed graph, $T$,
with a distinguished vertex, $r$, called the \emph{root}, such that,
for any vertex $v$ in $T$ that is not the root, there is exactly one path
from $r$ to $v$, e.g., see Tutte~\cite{tutte2001graph}.
That is, an arborescence is a graph-theoretic way of describing a
rooted tree, so that all the edges are going away from the root. 
In this paper, when we refer to a ``rooted tree'' it should be understood 
formally to be an arborescence.

We represent a rooted tree as $T=(V,E,r)$,  with 
a vertex set $V$, edge set $E$, and root $r \in V$. 
The \emph{degree} of a vertex in such a tree
is the sum of its in-degree and out-degree, and the degree of
a tree, $T$, is the maximum degree of all vertices in $T$.
So, an arborescence representing a binary tree would have degree 3.
Because of the motivating applications, e.g., from computational biology,
we assume in this paper that the trees we want to reconstruct 
have maximum degree that is bounded by a fixed constant, $d$.

Let us review a few terms regarding rooted trees.

\begin{definition}{(ancestry)} \label{def:ancestry}
Given a rooted tree, $T=(V,E,r)$,
we say $u$ is \emph{parent} of $v$ (and $v$ is a
\emph{child} of $u$) if there exists a directed edge $(u,v)$ in
$E$. The \emph{ancestor} relation is the transitive closure of the
parent relation, and the \emph{descendant} relation is the transitive
closure of the child relation.
We denote the number of descendants of vertex $s$ by $D(s)$.
A node without any children is called a \emph{leaf}.
Given two leaf nodes, $u$ and $v$ in $T$, their \emph{lowest common ancestor},
$lca(u,v)$, is the node, $w$ in $T$, that is an ancestor of both $u$ and $v$
and has no child that is also an ancestor of $u$ and $v$.
\end{definition}

We next define the types of queries we consider in this paper
for reconstructing a rooted tree, $T=(V,E,r)$.

\begin{definition}
\label{def:relative_distance}
A \emph{relative-distance query} for $T$ is a function,
$closer$, which takes three leaf nodes, $u$, $v$, and $w$ in $T$,
as input and returns the pair of nodes from the set, $\{u,v,w\}$,
that has
the lower lowest common ancestor.
That is, 
$closer(u,v,w)=(u,v)$ if $lca(u,v)$ is a descendant of $lca(u,w)=lca(v,w)$.
Likewise,
we also have that
$closer(u,v,w)=(u,w)$ if $lca(u,w)$ is a descendant of $lca(u,v)=lca(v,w)$,
and
$closer(u,v,w)=(v,w)$ if $lca(v,w)$ is a descendant of $lca(u,v)=lca(u,w)$.
\end{definition}

\cref{def:relative_distance} assumes $T$ is a binary tree (of degree~$3$). Note that
in this paper we restrict relative-distance queries to leaves, since
these represent, e.g., current species in the application of reconstructing
biological phylogenetic trees.

\begin{definition}
A \emph{path query} for $T$ is a function,
$path$, that takes two nodes, $u$ and $v$ in $T$, as input
and returns $1$ if there is a (directed) path from vertex $u$ to $v$, and
otherwise returns $0$. 
Also, for $u \in V$ and $W \subseteq V$,
we define 
$count(u,W) = \sum_{v \in W} path(u,v)$, which is the number of descendants of $u$ 
in $W$.
\end{definition}

We next study some preliminaries involving the structure of degree-$d$
rooted trees that will
prove useful for our parallel algorithms.

\begin{definition}\label{def:even_separator}
Let $T=(V,E,r)$ be a degree-$d$ rooted tree. We say that an
edge $e=(x,y) \in E$ is an \emph{even-edge-separator} if removing $e$ from
$T$ partitions it into two rooted trees,
$T^{\prime}=(V^{\prime},E^{\prime},y)$ and $T^{\prime \prime}=(V^{\prime
\prime},E^{\prime \prime},r)$, such that $ \frac{|V|}{d} \leq
|V^{\prime}|  \leq \frac{|V|  (d-1)}{d} $ and $ \frac{|V|}{d} \leq
|V^{\prime \prime}| \leq \frac{|V|  (d-1)}{d}$. (See
\cref{fig:subproblems}.)

\end{definition} 

\begin{lemma}\label{lem:has-separator}
  Every rooted tree of degree-$d$ has an even-edge-separator.
\end{lemma}

\begin{proof}
This follows from a result by Valiant~\cite[Lemma 2]{DBLP:journals/tc/Valiant81}.
\end{proof}
\begin{figure}[t]
  \centering
  \includegraphics[width=.9\columnwidth]{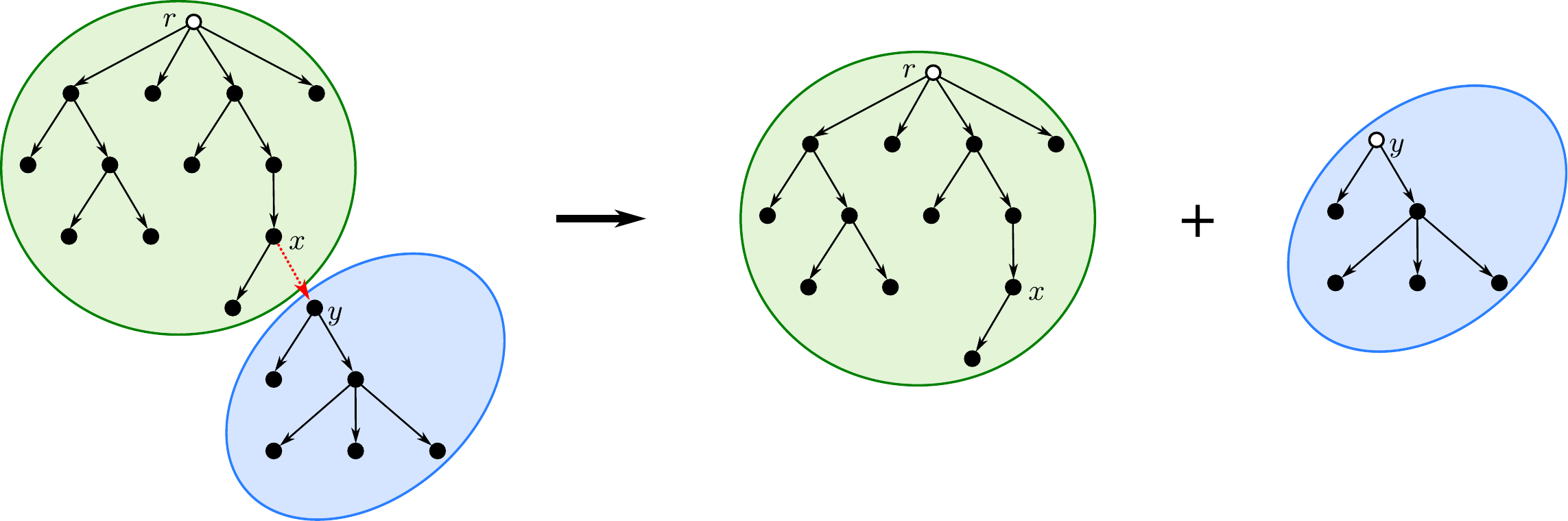}
  \caption{Illustration of a divide-conquer approach for trees. 
    The edge $(x,y)$ is an even-edge-separator. Note that the root of $T^{\prime \prime}$ is $r$, while $y$ becomes root of $T^{\prime }$.}
  \label{fig:subproblems}
\end{figure}

As we will see, this fact is useful for designing simple
parallel divide-and-conquer algorithms. 
Namely, if we can find an even-edge-separator, 
then we can cut the tree in two by removing that edge and recurse on the two 
remaining subtrees in parallel (see \cref{fig:subproblems}).

\section{Reconstructing Biological Phylogenetic Trees in Parallel}
\label{sec:phylo}
\emph{Relative-distance} queries model an experimental approach
to constructing a biological phylogenetic tree, e.g., where DNA sequences are compared
to determine which samples are the most similar~\cite{DBLP:journals/jal/KannanLW96}. 
In simple terms, pairs of DNA sequences that are closer to one another 
than to a third sequence are assumed to be from two species 
with a common ancestor that is more recent than the common ancestor of all three.
In this section,
for the sake of tree reconstruction, we assume the \emph{responder} has
knowledge of the absolute structure of a 
rooted binary phylogenetic tree; hence, each response to a
$closer(u,v,w)$ query is assumed accurate with respect to an unknown
rooted binary tree, $T$.  As in the pioneering work of Kannan {\it et al.}~\cite{DBLP:journals/jal/KannanLW96},
we assume the distance comparisons are accurate and consistent.
The novel dimension here is that we consider parallel algorithms for phylogenetic 
tree reconstruction.

As mentioned above, we consider relative-distance queries to occur between
leaves of a rooted binary tree, $T$.
That is, in our query model, 
the \emph{querier} has no knowledge of the internal nodes of $T$
and can only perform queries using leaves. 
Because $T$ is a binary phylogenetic tree, we may assume it is
a \emph{proper} binary tree, where each internal node in $T$ has exactly two children.

\subsection{Algorithm}
At a high level, our parallel 
reconstruction algorithm 
(detailed in \cref{phylo:overview})
uses a randomized divide-and-conquer approach, similar to \cref{fig:subproblems}.
In our case, however, the division process is random three-way split through a vertex separator,
rather than an edge-separator-based binary split.
Initially, all leaves belong to a single partition, $L$. 
Then two leaves, $a$ and $b$, are
chosen uniformly at random from $L$ and each 
remaining leaf, $c$, is queried in parallel
against them using relative-distance queries. 
Notice that the lowest common ancestor of $a$ and $b$ splits the tree into three parts.
Given $a$ and $b$,
the other leaves are split into three subsets ($R$, $A$, and $B$) 
according to their query result
(as shown in \cref{fig:phylo}(a)): 
\begin{itemize}
\item
$A$: leaves close to $a$, i.e., for which $closer(a,b,c)=(a,c)$
\item
$B$: leaves close to $b$, i.e., for which $closer(a,b,c)=(b,c)$
\item
$R$: remaining leaves, i.e., for which $closer(a,b,c)=(a,b)$
\end{itemize}
\begin{figure}[t]
\centering
\includegraphics[width=.7\linewidth,page=2]{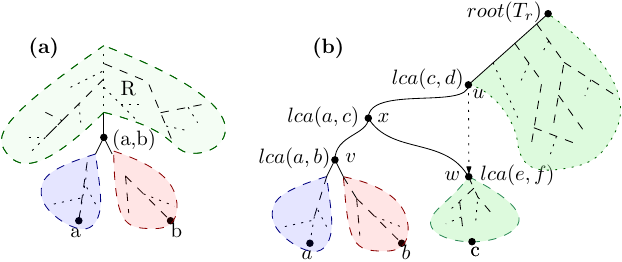}
\caption{\textbf{(a)} The subgroups leaves are split into. \textbf{(b)} The linking step attaching $T_a$, $T_b$ and $T_r$.}
\label{fig:phylo}
\end{figure}
We then recursively construct the trees in parallel:
$T_a$, for $A\cup\{a\}$;
$T_b$, for $B\cup\{b\}$;
and
$T_r$, for $R$.
The remaining challenge, of course, is to merge these trees to
reconstruct the complete tree, $T$.
The subtree of $T$ formed by subset $A \cup B$ is
rooted at an internal node, $v=lca(a,b)$; 
hence, we can create a new node, $v$, label it ``$lca(a,b)$'' and let $T_a$
and $T_b$ be $v$'s children.
If $R=\emptyset$, then we are done. 
Otherwise,
we need to determine the parent of $v$ in~$T$;
that is, we need to link $v$ into $T_r$ using function \textsf{link($v, T_r$)} (see \cref{phylo:link}).
\begin{algorithm}[bt!]
\caption{Reconstruct a binary tree of a set of leaves,~$L$.}
\label{phylo:overview}
 \SetKwFunction{FMain}{\textsf{reconstruct-phylogenetic}}
  \SetKwProg{Fn}{Function}{:}{}
  \SetKwProg{query}{query}{}{}
  \Fn{\FMain{$L$}}{
  \SetAlgoNoLine
\DontPrintSemicolon
\lIf{$|L|\leq 3$}{\KwRet{\rm the tree formed by querying $L$}}
Pick two leaves, $a,b \in L$, uniformly at random\;
\ForPar{each $c\in L$ s.t. $c\not=a,b$}{
Perform query $closer(a,b,c)$\;
}
Split the leaves in $L$ into $R$, $A$, and $B$ based on results\;

\SetKwBlock{Pardo}{parallel do}{}
\Pardo{
$T_a\gets$ \textsf{reconstruct-phylogenetic}$(A\cup \{a\}$)\; 
$T_b\gets$ \textsf{reconstruct-phylogenetic}$(B\cup \{b\})$\;
$T_r\gets$ \textsf{reconstruct-phylogenetic}$(R)$\;
}
Let $v$ be a new node, labeled ``$lca(a,b)$''\;
Set $v$'s left child to $root(T_a)$ and the right to $root(T_b)$\;
\lIf{$R=\emptyset$}{\KwRet{ \rm tree, $T_v$, rooted at $v$}}
\lElse{
\KwRet{\textsf{\rm link}$(v,T_r)$}
}
}
\end{algorithm}

 To identify the parent of $v$, in $T$, let us assume inductively that each 
internal node 
$u \in T_r$ has a label ``$lca(c,d)$'', 
since we have already recursively labeled each internal node in $T$.
Recall that $v$ is labeled with ``$lca(a,b)$''. The crucial observation is to note if
there exists an edge $(u\to w)$ in $T_r$, such that 
$u$ is labeled ``$lca(c,d)$'' and $closer(a,c,d) = (a,z)$ for $z\in\{c,d\}$, and 
w is either leaf~$z$ or an ancestor of $z$ labeled ``$lca(e,f)$'' with $closer(a,e,f)=(e,f)$, 
(See \cref{fig:phylo}(b)), then edge $(u\to w)$ must be where the parent of $v$ belongs in $T$, and if there
is no such edge, the parent of $v$ is the root of $T$ and the sibling of $v$
is the root of $T_r$.
We can determine the edge $(u \to w)$ in a single parallel round by performing the query,
$closer(a,c,d)$, for each each internal node $u \in T$ (where the label of $u$
is ``$lca(c,d)$''). It is also worth noting that if the oracle can identify cases where all three leaves share a single lca, simple modifications to \cref{phylo:overview} would enable it to handle trees of higher degree.

\begin{algorithm}[hbtp]
\caption{Link $v$ into the tree, $T_r$.}
\label{phylo:link}
 \SetKwFunction{FMain}{\textsf{link}}
  \SetKwProg{Fn}{Function}{:}{}
  \SetKwProg{query}{query}{}{}
  \Fn{\FMain{$v,T_r$}}{
  \SetAlgoNoLine
\DontPrintSemicolon
\ForPar{each internal node $u\in T_r$}{
Query $closer(a,c,d)$, 
where $u$'s label is ``$lca(c,d)$''\;
}
Let $(u\to w)$ be the edge in $T_r$ such that: \linebreak
$u$ is labeled ``$lca(c,d)$'' and $closer(a,c,d)$ is $(a,z)$, 
\linebreak
where $z\in\{c,d\}$, and\linebreak
$w$ is leaf~$z$, or an ancestor of $z$ with
label ``$lca(e,f)$'' and $closer(a,e,f)=(e,f)$\;
\uIf{no such edge $(u\to w)$ exists}{
	Let $g$ be a leaf from $root(T_r)$'s $lca$ label\;
	Create a new root node, labeled ``$lca(a,g)$'' with $v$ as left child and $root(T_r)$ as right child\;
	\KwRet{\rm the tree rooted at this new node}
}
Remove $(u\to w)$ from $T_r$\;
Create a new node, $x$, labeled ``$lca(a,z)$'', 
with parent~$u$, left child~$v$, and right child~$w$\;
\KwRet{$T_r$}
}
\end{algorithm}

\newcommand{\node}[1]{$n_{\text{#1}}$}
\newcommand{\splits}{N_{\textsc{splits}}}

\subsection{Analysis}
The correctness of our algorithm follows from the way relative-distance
queries always return a label for the lowest common ancestor for the two
closest leaves among the three input nodes.
Furthermore,
executing the
three recursive calls can be done in parallel, because $A\cup\{a\}$,
$B\cup\{b\}$, and $R$ form a partition of the set of leaves, $L$, and at
every stage we only perform relative-distance queries relevant to the respective partition.

\begin{theorem}
\label{thm:bio}
Given a set, $L$, of $n$ leaves in a proper binary tree, $T$,
such as a biological phylogenetic tree, we can reconstruct $T$ 
using relative-distance queries
with 
a round complexity, $R(n)$, that is $O(\log n)$ and a query complexity,
$Q(n)$, that is $O(n\log n)$, with high probability.
\end{theorem}

\begin{proof}
Because the recursive calls we perform in each call to
the reconstruct algorithm are done in parallel on a partition of
the leaf nodes in $L$, we perform $\Theta(n)$ work per round. 
Thus, showing that the
number of rounds, $R(n)$, is $O(\log n)$ w.h.p.~also implies that
$Q(n)$ is $O(n\log n)$ w.h.p.
To prove this, we show that each round in \cref{phylo:overview} 
has a constant probability of decreasing the problem size by at 
least a constant factor for each of its recursive calls.
For analysis purposes, we consider the left-heavy representation
of each tree, in which the tree rooted at the left child of any
node is always at least as big as the tree rooted at its right
child. 
(See \cref{fig:split}.)
Using this view, we can characterize when a partition
determines a ``good split'' and provide
bounds on the sizes of the partitions, as follows.

\begin{figure}[t]
\centering
\includegraphics[width=.6\linewidth,page=1]{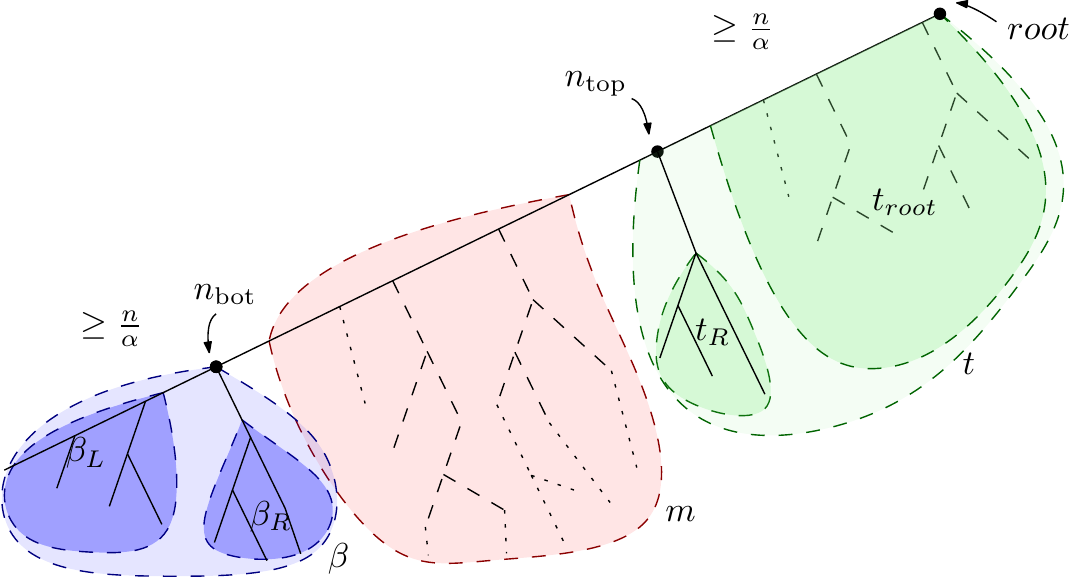}
\caption{A left-heavy tree drawing displaying node \node{top}, node \node{bot} and the relevant partitions}
\label{fig:split}
\end{figure}

\begin{lemma}\label{lm:phylo}
With probability of at least $\frac{\alpha-5}{2 \alpha ^2}$, a round of 
\cref{phylo:overview} 
decreases the problem size of any recursive
call by at least a factor of
$\frac{\alpha-1}{\alpha}$, for a constant $ \alpha > 5$,
thus experiencing a \emph{good split}.
\end{lemma}

\begin{proof}
Define the \emph{spine} to be all nodes on the path from the root to its left-most leaf in a left-heavy drawing.
 Let \node{bot} be the bottom-most node on the spine
that has at least $n/ \alpha$ of the nodes in its sub-tree. Conversely,
let \node{top} be the parent of the top-most spine node that has
at most $(1-1/ \alpha)\cdot n$ descendants.  (See \cref{fig:split}.)
Consider the three resulting trees obtained from separating at the
incoming edge to \node{bot} and the outgoing edge from \node{top}
to its left child. As shown in \cref{fig:split}, let $t$ be the
resulting tree retaining the root, $\beta$ the tree rooted at
\node{bot} and $m$ the tree between \node{top} and \node{bot}.
Within $\beta$ let $\beta_L$ and $\beta_R$ be the trees rooted at
the left and right child of \node{bot}.  Similarly, for $t$, let
$t_R$ be the tree rooted at the right child of \node{top}.
Finally, let $t_{root}$ be the remaining tree when cut at \node{top}.

Consider the size of tree $\beta$, $|\beta|$, since this is the first tree rooted in the spine with over $n/\alpha$ nodes, then $\beta_L$ must have had strictly under $n/\alpha$ nodes. Since the trees are in left-heavy order, $\beta_R$ can have at most as many nodes as $\beta_L$ so $\frac{n}{\alpha} \le |\beta|<\frac{2n}{\alpha}$.
Furthermore, we know that 
$|\beta|+|m|+|t_R|+|t_{root}|+1=n$. 
Due to the left-heavy order, $|t_R|\leq|\beta|+|m|$. 
By definition of \node{top}, it's necessary that $|t_{root}|<\frac{n}{\alpha}$, 
thus $2(|\beta|+|m|)+1>n-\frac{n}{\alpha}$ and $|\beta|+|m|> \left( \frac {\alpha-1}{2\alpha}\right)n-\frac{1}{2}$.
 Using the previous inequality and $|t|=n-|\beta|-|m|$, we find $|t|<\left(\frac{\alpha+1}{2\alpha}\right)n+\frac{1}{2}$.
Also, $|m|=n-|t|-|\beta|$, so 
$|m| > n-\frac{5n+\alpha n}{2\alpha}$. 

\begin{equation}
\label{eq:left_heavy}
    \frac{n}{\alpha} \le |\beta| <  \frac{2n}{\alpha}, \qquad  \frac{n}{\alpha} \le |t| < \left(\frac{\alpha+1}{2\alpha}\right)n+\frac{1}{2},  \qquad \left(\frac{\alpha-5}{2\alpha}\right)n < |m|\le \left(\frac{\alpha -2 }{\alpha}\right)n
\end{equation}
Picking a leaf from $\beta$ and another from $m$ guarantees that
$\beta \subseteq \left( A\cup\{a\} \right)$ and $t \subseteq R$. Thus, using \cref{eq:left_heavy}, each of the three sub-problem sizes, $|A\cup\{a\}|$, $|B\cup\{b\}|$, and $|R|$,  will be at most $\left(\frac{\alpha-1}{\alpha}\right)n$, when $\alpha > 5$. 
$Pr[\text{good split}]\geq Pr[\text{leaf in }\beta]\cdot Pr[\text{leaf in }m]$.
Asymptotically, $Pr[\text{leaf in }\beta]\approx Pr[\text{node in }\beta]$,
thus $Pr[\text{good split}]>\frac{\alpha-5}{2\alpha^2}$, which established
the lemma.
\end{proof}

Returning to the proof of \cref{thm:bio},
let $p=\frac{\alpha-5}{2\alpha ^2}$. From \cref{lm:phylo}, we expect it
will take $1/p$ rounds to obtain a good split. Every good split
will reduce the problem-size by at least a constant factor, 
$\frac{ \alpha -1}{\alpha }$. 
Thus, we are guaranteed to have just a single node left
after we get $\splits=\log_{\frac{\alpha}{\alpha-1}}(n)$
good splits.
Consider the geometric random
variable, $X_i$, describing the number of rounds required to obtain
the $i$-th good split, then $X=X_1+\ldots + X_{\splits}$ describes
the total number of rounds required by the algorithm. By linearity of 
expectation,
$E[X]=\frac{2\alpha^2}{(\alpha-5)}\cdot\splits=\frac{2\alpha^2}{(\alpha-5)}\cdot
\log_{\frac{\alpha}{\alpha-1}}(n)$.
Therefore, since $\alpha>5$ is a constant,
this already implies an expected
$O(\log n)$ rounds for \cref{phylo:overview}.
Moreover,
by a Chernoff bound for the sum of independent geometric random
variables (see~\cite{gt-adfai-02,DBLP:books/daglib/0012859}), 
$Pr\left[X>C\cdot E[X]\right]\leq e^{\frac{{-}(C-1)\cdot
p\cdot \splits}{5}}$
for any constant $C\geq3$ and constant
$\alpha >5$.
Thus, the probability that we take over $C\cdot E[X]$ rounds is
$O(1/n^{C-1})$. Therefore, by a union bound across the $n$
leaves, our algorithm completes in $O(\log n)$ rounds w.h.p. 
\end{proof}

\begin{corollary}\label{cor:bio_optimal}
\cref{phylo:overview} is optimal
when asking $\theta({n})$ queries per round.
\end{corollary}

The query complexity of \cref{phylo:overview} matches an $\Omega(n\log n)$ lower bound for $Q(n)$, due to 
Kannan {\it et al.}~\cite{DBLP:journals/jal/KannanLW96}. Besides, we need $\Omega(\log {n})$ rounds if we have $\theta({n})$ processors; hence, the round complexity of \cref{phylo:overview} is also optimal.

\let\oldnl\nl
\newcommand{\nonl}{\renewcommand{\nl}{\let\nl\oldnl}}

\section{Reconstructing Phylogenetic Trees from Path Queries}
\label{sec:path}
Let $T=(V,E,r)$ be a rooted (biological or digital) 
phylogenetic tree with fixed degree, $d$.
In this section, we show how a querier can reconstruct $T$ by issuing
$Q(n) \in O(n\log n)$ path queries in $R(n) \in O(\log n )$ 
rounds, w.h.p., where $n=|V|$. We provide a lower bound to prove that our algorithm is optimal in terms of query complexity and round complexity.
At the outset, the only thing we assume the querier knows is $n$ and $V$, that is,
the vertex set for $T$, and that the names of the nodes in $V$ are unique, i.e.,
we may assume, w.l.o.g., that $V=\{1,2,\ldots,n\}$.
The querier doesn't know $E$ or $r$---learning these is his goal.

\subsection{Algorithms}
We start by learning $r$, which we show can be done
via any maximum-finding algorithm in Valiant's parallel 
model~\cite{DBLP:journals/siamcomp/Valiant75}, which 
  only counts parallel steps involving comparisons.
The challenge, of course, is that the ancestor relationship
in $T$ is, in general, not a total order,
as required by a maximum-finding algorithm.
This does not actually pose a problem, however.

\begin{lemma}\label{lem:max-finding}
Let $A$ be a parallel maximum-finding algorithm 
in Valiant's model, with
$O(f(n))$ span and $O(g(n))$ work.
We can use $A$ to find the root, $r$ in a rooted tree $T=(V,E,r)$,
using $R(n)\in O(f(n))$ rounds and $Q(n)\in O(g(n)+n)$ total queries.
\end{lemma}

\begin{proof}
We pick an arbitrary vertex $v\in T$. In the first round, we perform
queries $path(u,v)$ in parallel for every other vertex $u\in V$ to
find $S$, the ancestor set for $v$. If $S = \emptyset$, then $v$ is the
root. Otherwise, we know all the vertices in
a path from root to the parent of $v$, albeit unsorted. 
Still, note that for $S$ the ancestor relation is a total order;
hence, we can simulate $A$ with path queries to resolve the
comparisons made by $A$. 
We have just a single round and $O(n)$
queries more than what it takes for $A$ to find the maximum. 
Thus, we can
find the root in $O(f(n))$ rounds and $O(g(n)+n)$ queries. 
\end{proof}

Thus, by well-known maximum-finding algorithms, e.g.,
see~\cite{DBLP:conf/stoc/ColeV86, DBLP:conf/conpar/ShiloachV81, DBLP:journals/siamcomp/Valiant75}:

\begin{corollary}
We can find $r$ of a rooted tree $T=(V,E,r)$ deterministically in  
$O(\log \log n)$ rounds and $O(n)$ queries.
\end{corollary}

Determining the rest of the structure of $T$ is more challenging, however.
At a high level, our approach to solving this challenge
is to use a separator-based divide-and-conquer
strategy, that is, find a ``near'' edge-separator
in $T$, 
divide $T$ using this edge, and recurse on the two remaining subtrees in parallel.
The difficulty, of course, is that the querier has no knowledge of the edges of $T$;
hence, the very first step, finding a ``near'' edge-separator, is a bottleneck
computation.
Fortunately, as we show in \cref{lem:exist-seperator},
if $v$ is a randomly-chosen vertex, then,
with probability depending on $d$,
the path from root $r$ to $v$
includes an edge-separator.

\begin{lemma} \label{lem:exist-seperator}
Let $T=(V,E,r)$ be a rooted tree of degree $d$ and let $v$ be a vertex chosen uniformly at random from $V$. Then, with probability at least $\frac{1}{d}$, an even-edge-separator is one of the edges on the path from $r$ to $v$.  
\end{lemma}

\begin{proof}
By \cref{lem:has-separator}, $T$ has an even-edge-separator. Let $e=(x,y)$ be an even-edge-separator for $T=(V,E,r)$ and let ${T^{\prime}=(V^{\prime},E^{\prime},y)}$ be the subtree rooted at $y$ when we remove $e$. 
Then, every path from $r$ to each $v \in V^{\prime}$ must contain $e$. 
By \cref{def:even_separator}, $T^{\prime}$ has at least $|V|/d$ vertices. 
Therefore, if we choose $v$ uniformly at random from $V$, then with probability $\frac{|V^{\prime}|}{|V|} \ge \frac{1}{d}$, 
the path from $r$ to $v$ contains $e$. 
\end{proof}

\begin{definition}{(splitting-edge)} \label{def: splitting_edge}
In a degree-$d$ rooted tree, an edge $(parent(s),s)$ is a splitting-edge if $ \frac{|V|}{d+2} \leq   \textit{}{D}(s) \leq \frac{|V|  (d+1)}{d+2}$, where $\textit{D}(s)$ is the number of descendants of $s$.
\end{definition}

Note that a degree-$d$ rooted tree $T$ always has a splitting-edge,
 as every even-edge-separator is also a splitting-edge and 
 by \cref{lem:has-separator}, it always has an even-edge-separator---a fact we use in our tree-reconstruction algorithm, which we describe next.
This recursive algorithm (given in pseudo-code 
in \cref{alg:rooted-reconstruct}),
assumes the existence of a randomized 
method, \textsf{find-splitting-edge}, which returns
a splitting-edge in $T$, with probability $\Omega(1/d)$, and otherwise returns 
$\mathit{Null}$.
Our reconstruction algorithm is therefore a randomized recursive algorithm that 
takes as input a set of vertices, $V$,  
with a (known) root vertex $r \in V$, and
returns the edge set, $E$, for $V$.
At a high level, our algorithm is to repeatedly call
the method, \textsf{find-splitting-edge}, until it returns a splitting-edge,
at which point we divide the set of vertices using this edge and recurse on
the two resulting subtrees.

\begin{algorithm}[hbt!] 
\caption{Reconstruct a rooted tree with path queries}
\label{alg:rooted-reconstruct}

\SetKwFunction{FMain}{\textsf{reconstruct-rooted-tree}}
\SetKwProg{Fn}{Function}{:}{}
\SetKwFor{ParQuery}{for each}{query in parallel}{}
\Fn{\FMain{$V,r$}}{
\SetAlgoNoLine
\DontPrintSemicolon
  $E \gets \emptyset$

  \If(\tcp*[h]{$g$ is a chosen constant}){$|V|\le g$}{
    \KwRet edges found by a quadratic brute-force algorithm    
  }
    \While{true}{
      Pick a vertex $v \in V$ uniformly at random\\
    	\ForPar{$z\in V$}{Perform query $path(z,v)$}
       Let $Y$ be the vertex set of the path from $r$ to $v$\\
      splitting-edge $\gets$ \textsf{find-splitting-edge}$(v, Y, V)$\\
      \If{splitting-edge $\neq$ $\mathit{Null}$}{
        $(u,w) \gets$ splitting-edge\\
        $E \gets E \cup \{(u,w)\}$\\
	\ForPar{$z\in V$}{Perform query $path(w,z)$}
        split $V$ into $V_1,V_2$ at $(u,w)$ using query results\;
        \SetKwBlock{Pardo}{parallel do}{}
        \Pardo{
            {$E \gets  E \cup 
            \textsf{reconstruct-rooted-tree}(V_1,w)$}

            {$E \gets E \cup 
            \textsf{reconstruct-rooted-tree}(V_2,r)$}
        }
        \KwRet $E$
      }
    }
}
\end{algorithm}

In more detail,
during each iteration of a repeating \textbf{while} loop, we choose a vertex 
$v \in V$ uniformly at random. 
Then, we find the vertices on the path from $r$ to $v$ and store them in 
a set, $Y$, using the fact that
a vertex, $z$, is on the path from $r$ to $v$ if and only if $path(z,v)=1$.  
Then, we attempt to find a splitting-edge using the function \textsf{find-splitting-edge} 
(shown in pseudo-code in Algorithm~\ref{alg:splitting-edge-finding}). 
If \textsf{find-splitting-edge} is unsuccessful, we give up on
vertex $v$, and restart the \textbf{while} loop with a new choice for $v$.
 Otherwise, \textsf{find-splitting-edge} succeeded and we cut the tree at the returned splitting-edge, $(u,w)$.
All vertices, $z \in V$, where $path(w,z)=1$ belong to the subtree rooted at $w$, thus belonging to $V_1$, whereas the remaining vertices belong to $V_2$ and the partition containing both $u$ and rooted at $r$.
Thus, after cutting the tree we recursively \textsf{reconstruct-rooted-tree} on $V_1$ and $V_2$.

\begin{algorithm}[hbt]
\caption{Finding a splitting-edge from vertex set, $Y$, on the path from vertex $v$ to the root $r$}
\label{alg:splitting-edge-finding}

  \SetKwFunction{FMain}{\textsf{find-splitting-edge}}
  \SetKwProg{Fn}{Function}{:}{}
  \Fn{\FMain{$v, Y, V$}}{
  \SetAlgoNoLine
\DontPrintSemicolon

    splitting-edge $\gets$ $\mathit{Null}$\\
    $m = C_1 \sqrt{|V|}$, $K = C_2 \log{|V|}$ \label{split_3} 
    
  \nonl  \SetKwProg{myalgo}{Phase 1:}{}{}
\myalgo{}{
     \If{$|Y| > |V|/K$}{
        $S \gets$ subset of $m$ random elements from $Y$
        
        $S \gets S \cup{\{v,r\}}$ \label{alg:split_6}

        \ForPar{each $s \in S$}{
            $X_s \gets$ subset of $K$ random elements from $V$
            
           Perform queries to find $count(s,X_s)$
        }

        \lIf{$\forall s \in S : count(s,X_s) < \frac{K}{d+1}$ }{\KwRet $\mathit{Null}$}    \label{split_10} 
          
          \lIf{$\forall s \in S : count(s,X_s) > \frac{K d}{d+1}$}{\KwRet $\mathit{Null}$}   \label{split_11}
          
            \If{$ \exists s \in S : \frac{K}{d+1} \le count(s,X_s) \le \frac{K d}{d+1}$}{
            \nonl	\KwRet \rm {\textsf{verify-splitting-edge}$(s,V)$}\label{split_12}
            }
            
         \ForPar{each $\{a,b\} \in S$}{ \label{alg:split_13}
            perform query $path(a,b)$
        }
        Find $w, z$ such that they are two consecutive nodes in the sorted order of $S$ such that  $count(w,X_w) > \frac{K d }{d+1}$ and $count(z,X_z) < \frac{K}{d+1}$ \label{split_15}\; 
        $Y \gets $  nodes from $Y$ in the path from $w$ to $z$ \label{split_16}
    }

       \SetKwProg{myalg}{Phase 2:}{}{}
\nonl \myalg{}{ 
        \lIf{$|Y| > |V|/ K $}{\KwRet{$\mathit{Null}$}} \label{split_17}
        \ForPar{each  $ s \in Y$}{
        $X_s \gets$ subset of $K$ random elements from $V$\\
        Perform queries to find $count(s,X_s)$
    }
    \If{$\exists s \in Y$\rm{ s.t. }$ \frac{K}{d+1} \leq count(s,X_s) \leq \frac{K d}{d+1}$}{
        \nonl    \KwRet \rm \textsf{verify-splitting-edge}$(s,V)$} 
            \label{split_21}
    }   
    }  
    \KwRet $\mathit{Null}$ \label{split_22}
    }
        
  \end{algorithm}

The main idea for our 
efficient tree reconstruction algorithm lies
in our \textsf{find-splitting-edge} method (see Algorithm~\ref{alg:splitting-edge-finding}),
which we describe next.
This method
takes as input the vertex $v$, the vertex set $Y$, (comprising 
the vertices on the path from $r$ to $v$), and the vertex set $V$. 
As we show, with probability depending on $d$, the 
output of this method is a splitting-edge; otherwise, the output is $\mathit{Null}$. 
Our algorithm performs a type of ``noisy''
search in $Y$ to either locate a likely splitting-edge or return $\mathit{Null}$ as 
an indication of failure.

Our \textsf{find-splitting-edge} algorithm consists of two phases.
We enter \textbf{Phase 1} if the size of path $Y$ is too big, i.e.,
$|Y| > |V|/K = \frac{|Y|}{C_2 \log {|V|} }$, where $C_2$ is a
predetermined constant and $K = C_2 \log |V|$. The purpose of this phase is either to pass
a shorter path including an even-edge-separator to the second phase
or to find a splitting-edge in this iteration.  The search on the
set $Y$ is noisy, because it involves random sampling. In particular,
we take a random sample $S$ of size $m = C_1 \sqrt {|V|}$ from path $Y$ (where $C_1$ is a predetermined constant). We include $r$ and $v$, the 
two endpoints of the path $Y$, to $S$. Then, we estimate the number of descendants of $s$,
$\textit{D}(s)$, for each $s \in S$.
To estimate this number for each $s \in S$, we take a random sample $X_s$ of $K$ elements from $V$ and we perform
queries to find $count(s,X_s)$.
Here, we use $m\cdot K \in O(\sqrt {|V|}  \log {|V|})$ queries in a
single round. Then, if all the estimates were less than $K/(d+1)$,
we return $\mathit{Null}$ as an indication of failure (we guess
that all the nodes on the path $Y$ have too few descendants to be
a separator). Similarly, if all the estimates were greater than
$\frac{K d}{d+1}$, we return $\mathit{Null}$ (we guess that all the
nodes on the path $Y$ have too many descendants to be a separator).
If there exists a node $s$ such that $\frac{K}{d+1} \le count(s,X_s)
\le \frac{K d}{d+1}$, we check if $s$ is a splitting-edge by
counting its descendants using a function, $\textsf{verify-splitting-edge}$. 
This function takes vertex $s$ and the full vertex set $V$ to return edge $($\textsf{find-parent}$(s,V),s)$ if $ \frac{|V|}{d+2} \leq count(s,V) \leq |V| \cdot \frac{d+1}{d+2}$ and return $\mathit{Null}$
otherwise.

If neither of these three cases happens, we perform
queries to sort elements of $S$ using a trivial quadratic work
parallel sort which takes $O(m^2) \in O(|V|)$ queries in a single
round. We know that two consecutive nodes $w$ and $z$ exist on the sorted
order of $S$, where $count(w,X_w)> \frac{K d}{d+1}$
and $count(z,X_z) < \frac{K}{d+1}$. We find all the nodes on $Y$
starting at $w$ and ending at $z$, and use this as our new path $Y$.

In \textbf{Phase 2}, we expect
a path of size under $|Y|/K$, we will later prove this is true with high probability.
 Otherwise, we just return
$\mathit{Null}$. In this phase, we estimate the number of
descendants much like we did in the previous phase, except the only
difference is that we estimate the number of descendants for all
the nodes on our new path $Y$. If there exists a node $s \in Y$ such that
$\frac{K}{d+1} \le count(s,X_s) \le \frac{K d}{d+1}$, we verify if
it is a splitting-edge, as described earlier.

Finally, let us describe how we find the parent of a node $s$ in $V$.
We first find, $Y$, the set of ancestors of $v$ in $V$ in parallel
using $|V|$ queries.  Let $x\succ y$ describe the total order of nodes in path $Y$,
where for any $ x,y \in Y : x \succ y$ if and only if $path(x,y)=1$.
The parent of $s$ is the lowest vertex on the path. Then, the key idea is
that if $|Y| \in O(\sqrt{|V|})$, we can sort them
using $O(|V|)$ queries. If the path is greater than this amount, we instead use $S$, a sample of size $O(\sqrt{|V|})$ from the path.
Next, we sort the sample to obtain $x_1 < \ldots < x_m$ for $S$ and then
 find all of the nodes in $Y$ which are less than the smallest sample
$x_1$. Finally, we replace $Y$ with these descendants of $x_1$ and repeat the whole
procedure again. We later prove that with high probability after two iterations of this
sampling, the size of the path is $O(\sqrt{|V|})$, allowing us to sort all nodes in $Y$ to return the minimum (see Function~\ref{alg:parent-finding}).

  \begin{function}[t]
    \caption{find-parent($s,V$)}\label{alg:parent-finding}
  \SetAlgoNoLine
\DontPrintSemicolon

    find $Y$, ancestor set of $s$ from $V$ using $|V|$ queries in parallel
    
    $m = C_1 \sqrt {|V|}$
     
     \For{$i \gets 1$ to $2$ $ \And |Y| > m$}{
             
             $S \gets$ random subset of $Y$ with $m$ elements

             sort $S$ as $x_1< ... < x_m$ using $O(m^2)$ queries in parallel

            find $Y^{\prime} = \{u \in Y \mid u \leq x_1 \} $ using $O(|Y|)$ queries in parallel, replace $Y$ with $Y^{\prime}$
     }
     \If{$|Y| \leq m$}{
        sort $Y$ using $O(m^2)$ queries in parallel
        
        \KwRet (minimum of this path)
        
        }
    \KwRet $\mathit{Null}$
    
    \end{function}

\subsection{Analysis}
The correctness of the algorithm follows from the fact that
our method first finds the root, $r$, of $T$ and then finds the parent of 
each other node, $v$ in $T$.

\begin{theorem}\label{thm:splitting-edge-finding_query_complexity}
Given a set, $V$, of nodes of a rooted tree, $T$, such as a biological
or digital phylogenetic tree, with degree bounded
by a fixed constant, $d$, we can construct $T$
using path queries with round complexity, $R(n)$, that is $O(\log n)$ and 
query complexity, $Q(n)$, that is $O(n\log n)$, with high probability.
\end{theorem}

Our proof of \cref{thm:splitting-edge-finding_query_complexity} begins with
\cref{lem:elements_scattered}.

\begin{lemma}\label{lem:elements_scattered}
In a rooted tree, $T=(V,E,r)$, let $Y$ be a (directed) path, where 
$ |Y| > m = C_1 \sqrt{|V|}$. If we take a sample, $S$, of $m$
elements from $Y$, then with probability $1 - \frac{1}{|V|}$, every
two consecutive nodes of $S$ in the sorted order of $S$ are within
distance $O\left(\frac{|Y| \log {|V|}}{\sqrt{|V|}}\right)$ from
each other in $Y$.
\end{lemma}

\begin{proof}
 Note that some nodes of $Y$ may be picked more than once as we
 pick $S$ in parallel. 
Divide the path $Y$ into $\frac{\sqrt{|V|}}{\log |V|}$ 
equal size sections (the difference between the size of any two
 sections is at most 1). 
For each $1 \le i \le {\frac{\sqrt{|V|}}{\log |V|}}$, 
let $A_i$ be the subset of $S$ lying in the
 $i\textsuperscript{th}$ section of $Y$. (See \cref{fig:lem_scattered}.)
It is clear that each node $s \in S$ ends up in 
section $i$ with probability $\frac{\log {|V|}}{\sqrt{|V|}}$,
 and therefore, for each $1 \le i \le {\frac{\sqrt{|V|}}{\log |V|}}$,
 ${E}\left[|A_i|\right] = C_1 \log {|V|}$. Thus, using standard a 
 Chernoff bound, $\Pr\left[|A_i| = 0\right] < \frac{1}{|V|^2}$ for any constant
 $C_1 > 6 \ln {10}$. Using a union bound, all the sections are
 non-empty with probability at least $1 - \frac{1}{|V|}$. Hence,
 the distance between any two consecutive nodes of $S$ from each
 other in $Y$ is at most $\frac{2 |Y| \log {|V|} }{\sqrt {|V|}}$.
\end{proof}

\begin{figure}[t]
  \centering
  \includegraphics[scale=.9]{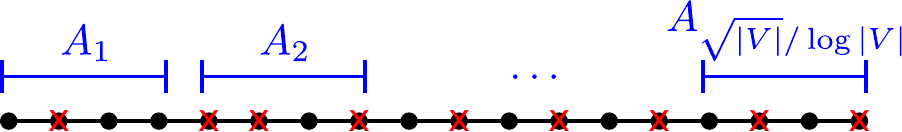}
  \caption{Illustration of how scattered sample $S$ is on path $Y$. The $i\textsuperscript{th}$ blue interval represents the $i\textsuperscript{th}$ section of $Y$, the black dots correspond to the nodes on the path $Y$, and red crossed marks represent elements of $S$. }
  \label{fig:lem_scattered}
\end{figure}

\cref{lem:elements_scattered} allows us to analyze the find-parent method, as follows.

\begin{lemma}\label{lem:parent_complexity}
  The \textsf{find-parent}$(s,V)$ method 
outputs the parent of $s$ with probability at least $1-2/|V|$, with $Q(n) \in O(n)$ and $R(n) \in O(1)$.
\end{lemma}

\begin{proof}
  The \textsf{find-parent} method succeeds if, after the \textbf{for} loop, the size of the set of remaining ancestors of $s$, $Y$, is $|Y|\le m$, so it is enough to show that this occurs with probability at least $1-2/|V|$. By \cref{lem:elements_scattered}, the size of $Y$ at the end of the first iteration is $|Y|\in O(m\log |V|)$, with probability at least $1-1/|V|$. Similarly, a second iteration, if required, further reduces the size of $Y$ into $|Y|\in o(m)$, with probability at least $1-1/|V|$. Thus, by a union bound, the probability of success is at least $1-2/|V|$.

  The query complexity can be broken down as follows, where $m\in O(\sqrt{|V|})$: 
  \begin{enumerate}
    \item $O(|V|)$ queries in 1 round to determine the ancestor set, $Y$, of $s$.
    \item $O(m^2) + O(|Y|)\in O(|V|)$ queries in 2 rounds for each of the (at most) 2 iterations performed in \textsf{find-parent}, whose purpose is to discard non-parent ancestors of $s$ in $Y$. 
    \item $O(m^2)\in O(|V|)$ to find, in 1 round, the minimum among the remaining ancestors of $Y$ (at most $m$ w.h.p.). If $|Y|>m$, then no further queries are issued.
  \end{enumerate}
  In total, the above amounts to $Q(n) \in O(n)$ and $R(n) \in O(1)$.
\end{proof}

We next analyze the \textsf{find-splitting-edge} method.

\begin{lemma}\label{lem:expected_d_times}
Any call to \textsf{find-splitting-edge} returns true
with probability $\frac{1}{2 d}$; hence 
\cref{alg:rooted-reconstruct} calls 
\textsf{find-splitting-edge} $O(d)$ times in expectation.
\end{lemma}

\begin{proof}
By \cref{lem:exist-seperator}, we know that if we pick a vertex $v$, 
uniformly at random, then with probability $\frac{1}{d}$, 
an even-edge-separator lies on the path from $r$ to $v$. 
We show that if there is such an even-edge-separator (Definition~\ref{def:even_separator}) on that path, \textsf{find-splitting-edge($v,Y,V$)} returns a 
splitting-edge (Definition~\ref{def: splitting_edge}) with probability at least $\frac{1}{2}$, and otherwise returns $\mathit{Null}$. 
Using an intricate Chernoff-bound analysis (see \cref{appendix:chernoff}), we can prove
that there exists a constant $C_2 > 0$, as used in line~\ref{split_3} of 
\cref{alg:splitting-edge-finding} such that the following probability bounds always hold:
\begin{equation}\label{eq:1}
  \begin{cases}
        
        \Pr\left(count(s,X_s) \geq \frac{K}{d+1}\right) \geq 1 - \frac{1}{|V|^2} & \mbox{if }count(s,V) \geq \frac{|V|}{d}, \\ 
        \Pr\left(count(s,X_s) \leq K  \frac{d}{d+1}\right) \geq 1 - \frac{1}{|V|^2} & \mbox{if } count(s,V) \leq |V|  \frac{d-1}{d},\\
               \end{cases}
 \end{equation}
 \begin{equation}\label{eq:2}
  \begin{cases}
 
        \Pr\left(count(s,X_s) < \frac{K}{d+1}\right) \geq 1 - \frac{1}{|V|^2} & \mbox{if }count(s,V) < \frac{|V|}{d+2}, \\
        \Pr\left(count(s,X_s) > K  \frac{d}{d+1}\right) \geq 1 - \frac{1}{ |V|^2} & \mbox{if }count(s,V) > |V|  \frac{d+1}{d+2} \\
        
        \end{cases}
 \end{equation}
 
It is clear that we either return a splitting-edge or $\mathit{Null}$ when passing through \textsf{verify-splitting-edge}. We break the probability of returning $\mathit{Null}$ according to the phases. We call a vertex $v$ \textbf{ineligible} if $count(v,V) < \frac{|V|}{d+2}$ or $count(v,V) > \frac{|V| (d+1)}{d+2}$ ($(parent(v),v)$ is not a splitting-edge). On the other hand, we call vertex $v$ \textbf{candidate} if after estimating the number of its descendants: $ \frac{K}{d+1} \leq count(v,X) \leq \frac{K  d}{d+1}$. Let $(a,b)$ be an even-edge-separator on path $Y$.

\textbf{Phase 1:}
    \begin{itemize}
        \item lines~\ref{split_10},\ref{split_11}: 
By \cref{def:even_separator}, $ \frac{|V|}{d} \leq count(b,V) \leq \frac{|V|  (d-1)}{d}$. We add $\{r,v\}$ to $S$ in line~\ref{alg:split_6} of the algorithm (the two endpoints of path $Y$). Notice that $\frac{|V|}{d} \leq count (b,V) \leq count(r,V)$ and that $ count(v,V) \leq count (b,V) \leq \frac{|V| (d-1)}{d} $. So, by \cref{eq:1}, with probability at least $1 - \frac{2}{|V|^2}$, $count(r,X_r) \geq \frac{K}{d+1}$ and $count(v,X_v) \leq \frac{K d}{d+1}$, and consequently, we don't return $\mathit{Null}$ in lines~\ref{split_10},\ref{split_11} of the algorithm.
        
        \item line~\ref{split_12}:
        \cref{eq:2} shows that an ineligible node is not a candidate with probability $1 - \frac{1}{|V|^2}$. Thus, by a union bound, none of our candidates is ineligible in line~\ref{split_12} with probability at least $1 - \frac{|S|}{|V|^2}$. Moreover, if there exists a candidate $s$ in $S$, the algorithm outputs a splitting-edge $(parent(s), s)$ with probability at least $1-\frac{2}{|V|}$, by \cref{lem:parent_complexity}.

        \item lines~\ref{split_15},\ref{split_16}: Let us partition $S$ into $S_l$ and $S_r$ (see \cref{fig:trim_path}), as follows:
        \begin{equation*}
    S_l = \{ s \in S \mid count(s,V) > count(b,V) \}, \qquad S_r = \{ s \in S \mid count(s,V) < count(b,V) \}
\end{equation*}
        Then, by definition of $b$:
        \begin{align*}
         & \forall s \in S_l: & count(s,V) > |V|/d, \qquad
        &   \forall s \in S_r: & count(s,V) < |V|(d-1)/d.
        \end{align*}
        and thus, by \cref{eq:1} and a union bound, we have with probability at least $1 - \frac{|S|}{|V|^2}$:
        \begin{align*}
          & \forall s \in S_l: & count(s,X_s) \ge K/(d+1), \qquad
          & \forall s \in S_r: & count(s,X_s) \le Kd/(d+1). 
        \end{align*}
        Finally, since $S$ does not contain any candidate nodes (otherwise we would have picked them in line~\ref{split_12}), the above inequalities imply that:
        \begin{align*}
          & \forall s \in S_l: & count(s,X_s)  > Kd/(d+1), \qquad & \forall s \in S_r:
          &  count(s,X_s) < K/(d+1).
        \end{align*}
        Therefore, $w\in S_l$ and $z\in S_r$, which implies that the subpath from $w$ to $z$ in $Y$ must include vertex $b$. This means that with probability $1 - O(\frac{1}{|V|})$, we either find a splitting-edge in this phase or pass $b$ to the next phase.
    \end{itemize}
Now, consider \textbf{Phase 2:}
    \begin{itemize}
        \item line~\ref{split_17}: 
        Here, if $|Y| > |V|/K $, then we have passed through Phase 1. Using \cref{lem:elements_scattered}, we know that since $S$ was a sample of size $C_1 \sqrt{|V|}$ from $Y$, with probability $1 - \frac{1}{|V|}$  the distance between any two consecutive nodes of $S$ in $Y$ was $O(|Y| \frac{\log {|V|}}{\sqrt{|V|}}) \in O(\sqrt{|V|} \log {|V|})$. Thus, the size of path $Y$ after passing through line~\ref{split_15} is at most $O(\sqrt{|V|} \log{|V|})$. Thus, the probability of returning $\mathit{Null}$ in line~\ref{split_17} is at most $\frac{2}{|V|}$.

        \item line~\ref{split_21}: By \cref{eq:2}, with probability at least $1 - \frac{|Y|}{|V|^2}$, no ineligible node is between candidate set. Besides, for a candidate node $s$, the algorithm outputs a splitting-edge $(parent(s), s)$ with probability at least $1-\frac{2}{|V|}$, by \cref{lem:parent_complexity}.
        
        \item line~\ref{split_22}: The probability that we return $\mathit{Null} $ here is equal to the probability that our candidate set in line~\ref{split_21} is empty. By \cref{eq:1}, with probability at least $1 - \frac{2}{|V|^2}$, $b$ is between candidates at line~\ref{split_21} and candidate set is non-empty.
        Thus, the total probability of failing to return a splitting-edge in this phase is at most $O(\frac{1}{|V|})$.

    \end{itemize}

Therefore, for $|V|$ greater than the chosen constant $g$, the probability of returning $\mathit{Null}$ in the existence of an even-edge-separator is at most $ O(\frac{1}{|V|}) \leq 1/2$. Thus, the probability of returning a splitting-edge in any call to \textsf{find-splitting-edge} is at least $\frac{1}{2d}$.
  \end{proof}

\begin{figure}
  \centering
  \includegraphics[scale=.9]{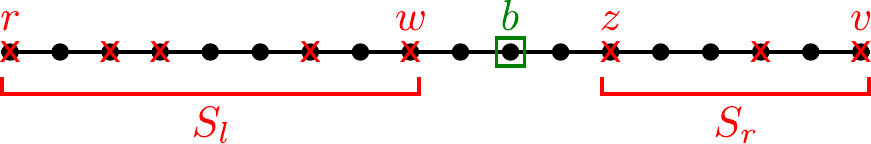}
  \caption{Illustration of the path reduction in Phase 1 of \textsf{find-splitting-edge}. At the end of this phase, the path $Y$ is trimmed down into the subpath consisting of the nodes between $w$ and $z$, which contains $b$ w.h.p.}
  \label{fig:trim_path}
\end{figure}

\begin{lemma}\label{lem:splitting_edge_query_complexity}
The subroutine \textsf{find-splitting-edge}$(v,Y,V)$ has query complexity, $Q(n)$, that is $O(|V|)$, and round complexity, $R(n)$, that is $O(1)$.
\end{lemma}

\begin{proof}
  The queries done by \textsf{find-splitting-edge}$(v,Y,V)$, in the worst case, can be broken down as follows, where $m=O(\sqrt{|V|})$:
  
\textbf{Phase 1:} A total of $O(|V|)$ queries in $O(1)$ rounds, consisting of:
    \begin{itemize}
      \item $O(mK) \in O(\sqrt{|V|}\log |V|)$ queries in one round for estimating the number of descendants for the $m$ samples.
      \item $O(m^2) \in O(|V|)$ queries in one round for sorting the $m$ samples.
      \item $O(|Y|) \in O(|V|)$ queries in a round to find the subpath of $Y$ that is the input for Phase 2.
    \item $O(|V|)$ queries in $O(1)$ rounds for determining the parent of $s$ (see \cref{lem:parent_complexity}).
    \end{itemize}

\textbf{Phase 2:} If we enter this phase, it spends $O(|V|)$ queries in $O(1)$ rounds:
    \begin{itemize}
      \item $|Y| \cdot K \in O(|V|)$ queries in one round to evaluate \textsf{count($s, X_s$)} for each $s \in Y$.
      \item $O(|V|)$ queries in one round to find the number of descendants of node $s$.
      \item $O(|V|)$ queries in $O(1)$ rounds to determine the parent of $s$ (see \cref{lem:parent_complexity}).
    \end{itemize}
  Overall, the above break down amounts to $Q(n) \in O(n)$ and $R(n) \in O(1)$.
  \end{proof}

Now, recall \cref{thm:splitting-edge-finding_query_complexity}:
\textit{
Given a set, $V$, of nodes of a rooted tree, $T$, such as a biological
or digital phylogenetic tree, with degree bounded
by a fixed constant, $d$, we can construct $T$
using path queries with round complexity, $R(n)$, that is $O(\log n)$ and 
query complexity, $Q(n)$, that is $O(n\log n)$, with high probability.
}

\begin{proof}
  The expected query complexity $Q(n)$ of \cref{alg:rooted-reconstruct} is dominated by the two recursive calls $\left(Q\left(\frac{n}{d+2}\right)\text{ and }Q\left(\frac{n(d+1)}{d+2}\right)\right)$ and the calls to \textsf{find-splitting-edge}. By \cref{lem:expected_d_times}, we call \textsf{find-splitting-edge} an expected $O(d)$ times, incurring a cost of $O\left(dn\right) \in O(n)$ path queries in $O(d) \in O(1)$ rounds (see \cref{lem:splitting_edge_query_complexity}).
  Thus, $Q(n)$ and $R(n)$ are:
 \[
  Q(n) = Q\left(\frac{n}{d+2}\right) + Q\left(\frac{n(d+1)}{d+2}\right) + O\left(n \right) ,
  \] 
  \[ 
  R(n) = \max \left( R\left(\frac{n}{d+2}\right), R\left(\frac{n(d+1)}{d+2}\right) \right) + O(1)
 \]
which shows it needs $Q(n) \in O(n \log n)$ and $R(n) \in O( \log n)$ in expectation. To prove the high probability results, note that the main algorithm is a divide-and-conquer
algorithm with two recursive calls per call; hence, it can be modeled with
a recursion tree that is a binary tree, $B$, with height
$h=O(\log_{\frac{d+2}{d+1}} n)=O(\log n)$. 
For any root-to-leaf path in $B$, the time taken can modeled as
a sum of independent random variables,
$X=X_1+X_2+\cdots+X_h$, where each $X_i$ is the number of calls to
\textsf{find-splitting-edge} (each of which uses $O(|V|)$ queries in $O(1)$ rounds)
required before it returns true,
which is a geometric random variable with parameter $p=\frac{1}{2 d}$.
Thus, by a Chernoff bound for sums of independent geometric random variables
(e.g., see~\cite{gt-adfai-02,DBLP:books/daglib/0012859}), 
the probability that $X$ is more than
$O(d\log_{\frac{d+2}{d+1}} n)$ is at most $1/n^{c+1}$, for any given constant $c\ge 1$.
The theorem follows, then, by a union bound for the $n$ root-to-leaf paths in $B$.
\end{proof}

\subsection{Lower Bound}\label{sec:lower-bound}
We establish the following simple lower bound,
which extends and corrects lower-bound proofs
of Wang and Honorio~\cite{DBLP:conf/allerton/WangH19}.

\begin{figure}[b!]
  \centering
  \includegraphics[width=0.45\linewidth]{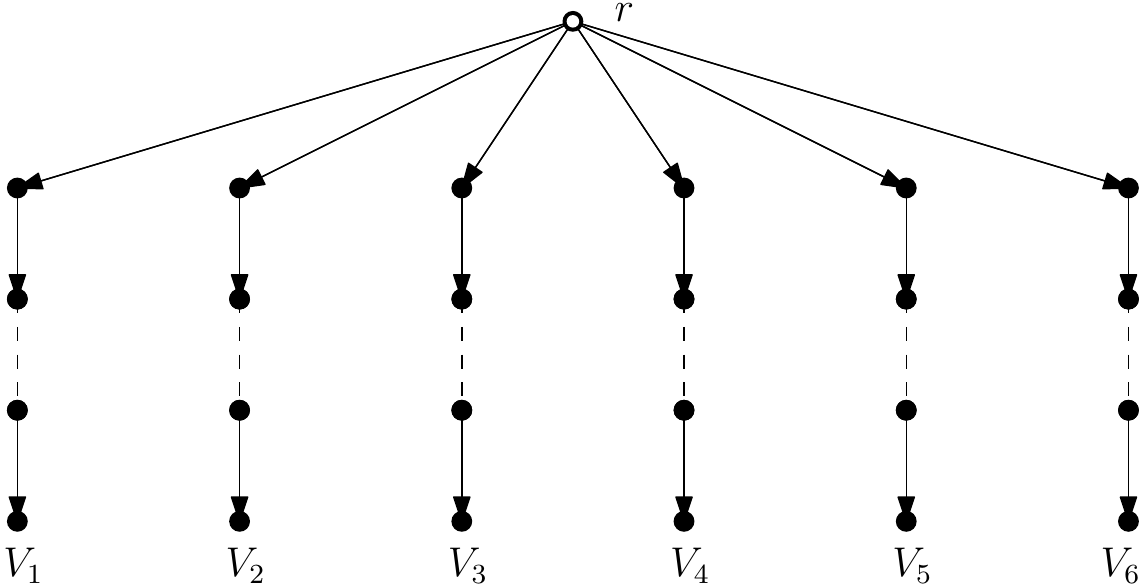}
  \caption{Illustration of the $\Omega(dn+n\log n)$ lower bound 
    for path queries in directed rooted trees
    (shown for $d=6$).
}
  \label{fig:lower_bound}
\end{figure}

 \begin{theorem}\label{thm:lower-bound}
Reconstructing an $n$-node, degree-$d$ tree requires
$\Omega(dn+n \log n)$ path queries.
This lower bound holds for the worst case of a deterministic algorithm and for
the expected value of a randomized algorithm.
\end{theorem}
\begin{proof}
Consider an $n$-node, degree-$d$ tree, $T$, as shown in \cref{fig:lower_bound},
which consists of a root, $r$, with $d$ children, each of which is the root
of a chain, $T_i$, of at least one node rooted at a child of $r$.
Since a querier, Bob,
can determine the root, $r$, in $O(n)$ queries anyway, let us assume
for the sake of a lower bound 
that $r$ is known; hence, no additional information is gained by path queries
involving the root.
Let us denote the vertices in chain $T_i$ as $V_i$.
In order to reconstruct $T$, Bob must determine 
the nodes in each $V_i$ and must also determine their order in $T_i$.
For a given path query,
$path(u,v)$, say this query is \emph{internal} if $u,v\in V_i$, for some
$i\in[1,d]$, and this query is \emph{external} otherwise.
Note that even if Bob knows the full structure of $T$ except for 
a given node, $v$, he must perform at least $d-1$ external queries in the worst case,
for a deterministic algorithm, or $\Omega(d)$ external queries in expectation, for
a randomized algorithm, in order to determine the chain, $T_i$, to which
$v$ belongs.
Furthermore, the result of an (internal or external) query, $path(u,v)$,
provides no additional information for a vertex $w$ distinct from $u$ and $v$
regarding the set, $V_i$, to which $w$ belongs.
Thus, Bob must perform $\Omega(d)$ external queries for each vertex $v\not=r$,
i.e., he must perform $\Omega(dn)$ external queries in total.
Moreover, note that the results of external queries involving a vertex, $v$,
provide no information regarding the location of $v$ in its chain, $T_i$.
Even if Bob knows all the vertices that comprise each $V_i$, he must determine the
ordering of these vertices in the chain, $T_i$, in order to reconstruct $T$.
That is, Bob must \emph{sort} the vertices in $V_i$ using a comparison-based
algorithm, where each comparison is an
internal query involving two vertices, $u,v\in V_i$.
By well-known sorting lower bounds (which also hold in expectation for randomized
algorithms), e.g., see~\cite{gt-adfai-02, DBLP:books/daglib/0023376},
determining the order of the vertices in each $T_i$ requires
$\Omega(|V_i| \log |V_i|)$, as one of the chain can be as great as $n-d$ vertices, then he needs $\Omega(n \log n)$ internal queries. 
\end{proof}

\begin{corollary}\label{cor:lower_bound_time}
\cref{alg:rooted-reconstruct} is optimal for bounded-degree trees 
when asking $\theta({n})$ queries per round.
\end{corollary}

The query complexity of Algorithm~\ref{alg:rooted-reconstruct} matches the lower bound provided by \cref{thm:lower-bound} when $d$ is constant. Besides, we need $\Omega(d + \log {n})$ rounds if we have $\theta({n})$ processors; hence, the round complexity of \cref{alg:rooted-reconstruct} is also optimal.
\section{Experiments}
\label{sec:exp}
Given that both of our algorithms are randomized and perform optimally
with high probability, we carried out experiments to analyze the
practicality of our algorithms
and compare their performance
with the best known algorithms for reconstructing
rooted trees. \footnote{The complete source code for our experiments, including the implementation of our algorithms and the algorithms we compared against, is available at {{\href{https://github.com/UC-Irvine-Theory/ParallelTreeReconstruction}{\textsf{github.com/UC-Irvine-Theory/ParallelTreeReconstruction}}}}~.}
\begin{figure}[b!]
\centering
\includegraphics[width=.7\linewidth,trim=0 1mm 0 9mm,clip]{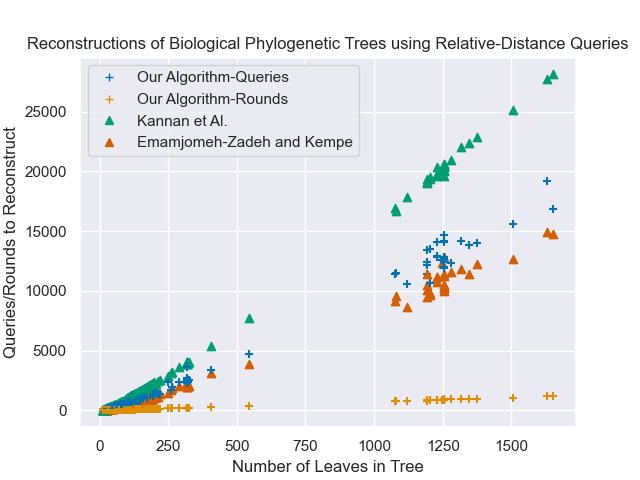}
\caption{A scatter plot showing the number of queries and rounds for each of the three tree reconstruction algorithms for real trees from TreeBase. Since our algorithm is parallel, we include round complexity to serve as a comparison for the sequential complexity.}
\label{fig:phyloexp}
\end{figure}
\subsection{Reconstructing Biological Phylogenetic Trees}
In order to assess the practical performance of \cref{phylo:overview}, 
we performed experiments using synthetic phylogenetic trees and
real-world biological phylogenetic trees from the phylogenetic library
TreeBase~\cite{piel2009treebase}, which is a database of biological phylogenetic trees,
comprising over 100,000 distinct taxa in total.

We implemented an oracle
interface, instantiated it
with the relevant trees, and implemented our algorithm along with two other
phylogenetic tree reconstruction algorithms that use relative-distance queries. 
The first is by Emamjomeh-Zadeh and
Kempe~\cite{DBLP:conf/soda/Emamjomeh-Zadeh18}, which is a randomized sequential divide-and-conquer 
algorithm.
The second is by Kannan~{\it et al.}~\cite{DBLP:journals/jal/KannanLW96}, where they use a 
sequential deterministic procedure reminiscent of insertion-sort.
All three algorithms achieve the optimal asymptotic
query complexity of $\Theta(n\log n)$ in
expectation.

\subparagraph*{Real Data.}
We instantiated our oracles
with 1,220 biological phylogenetic trees from 
the TreeBase collection and used them to run all three algorithms. The results, shown in \cref{fig:phyloexp}, suggest that our algorithm
outperforms the algorithm by Kannan~{\it et al.}, both in
terms of its round complexity and query complexity. 
However our algorithm almost matches Emamjomeh-Zadeh and Kempe's in terms of total queries and we believe the small difference is a direct result of the cost incurred while parallelizing the link step of \cref{phylo:overview}. It remains clear that \cref{phylo:overview} outperforms the two other algorithms when considering the parallel speed-up.

\begin{figure}[tb]
\centering
\includegraphics[width=.72\linewidth,trim=0 2mm 0 9mm,clip]{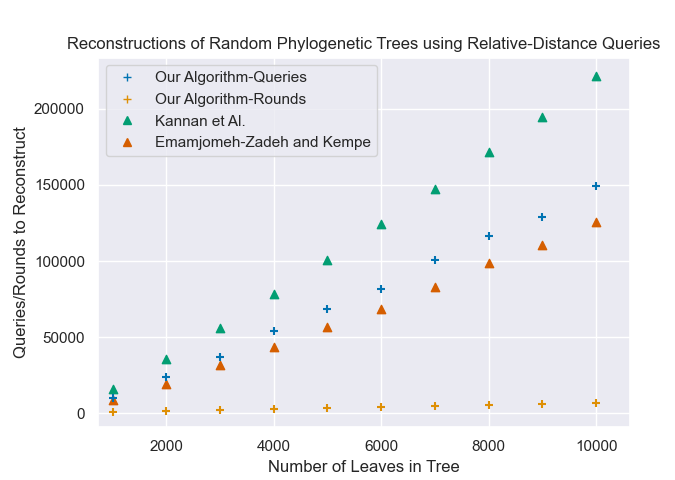}
\caption{A plot showing the average number of queries and rounds for each of the three tree reconstruction algorithms. Each data point represents the average for 10 randomly generated trees. }
\label{fig:phyloexpr}
\end{figure}

\subparagraph*{Synthetic Data.} We also tested this algorithm using synthetic data and found similar results, detailed in \cref{fig:phyloexpr}. We detail the method used to generate these random tree instances in \cref{sec:recpathexp}, however, given our algorithm's strict focus on biological phylogenetic trees, we use only full binary trees, where each internal node has exactly two children.

\subsection{Reconstructing Phylogenetic Trees from Path Queries}
\label{sec:recpathexp}
To assess the practical performance of our method for reconstructing
(biological and digital) phylogenetic trees from path queries,
we performed experiments using both synthetic and real data to compare 
our algorithm with the algorithm by Wang and Honorio~\cite{DBLP:conf/allerton/WangH19}, which is the best known reconstruction algorithm for phylogenetic trees from path queries.
Our experimental results provide evidence that \cref{alg:rooted-reconstruct} provides significant
parallel speedup,
while simultaneously improving the total number of queries.

\subparagraph*{Synthetic Data.}~In order to generate random instances of trees with maximum degree, $d$,
we synthesized a data set of random degree-$d$
trees of $n$ nodes for different values of $n$ and $d$. 
To generate a random tree, $T$, for a given $n$ and $d$,
we first generated a random Prüfer sequence \cite{prufer1918neuer} of a labeled
tree, which defines a unique sequence associated with that tree, and 
converted it to its associated tree. 
In particular, each $n$-node
tree $T= (V,E)$ has a unique code sequence 
$s_1, s_2, \dots , s_{n-2}$, where 
$s_i \in V$ for all $ 1\le i \le n-2$ and every node
$v_i \in V$ of degree $d_i$ appears exactly $d_{i}-1$ times in this
sequence. 
Therefore, in order to generate random degree-$d$ rooted
trees we generate a random Prüfer sequence while imposing conditions
that: (i) all vertices appear at most $d-1$ times in the code and
(ii) there is at least one node such that it appears exactly $d-1$
times. 
Converting each such sequence to its associated tree gave us
a random degree-$d$ tree instance that we used in our experiments.

Since our parallel reconstruction algorithm using
path queries is parameterized by a constant, $C_2$,  we ran our algorithm using different values for $C_2$. The constant $C_2$ controls sample size from $V$ used to
estimate the number of descendants of a node. Furthermore, to reduce noise from randomization, each data-point will be averaged for 3 runs on 10 randomly generated trees. In \cref{fig:fig_diff_n}, we compare our algorithm's rounds and total number of queries with the one by Wang and Honorio~\cite{DBLP:conf/allerton/WangH19}, for fixed degree trees $d=5$ and varying tree-sizes. 
These results provide empirical evidence that
our algorithm provides a noticeable speedup
in parallel round complexity while also outperforming the algorithm
by Wang and Honorio~\cite{DBLP:conf/allerton/WangH19} 
in total number of queries.

\begin{figure}[t]
\centering
\begin{subfigure}{.49\textwidth}
  \centering
  \includegraphics[width=\linewidth,trim= 0 4mm 0 3mm,clip]{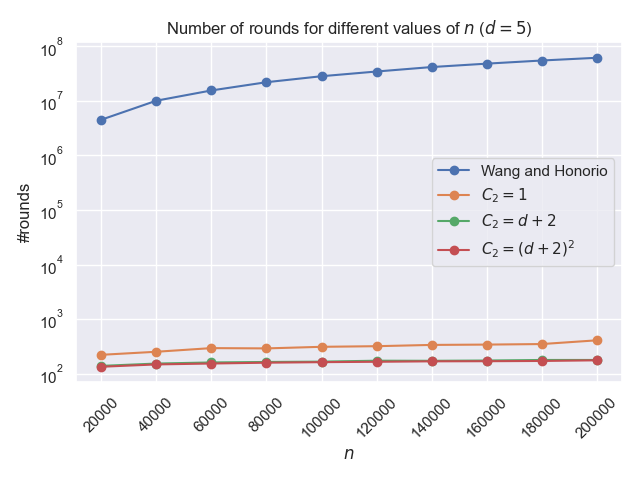}  
\end{subfigure}
\begin{subfigure}{.49\textwidth}
  \centering
  \includegraphics[width=\linewidth,trim= 0 4mm 0 3mm,clip]{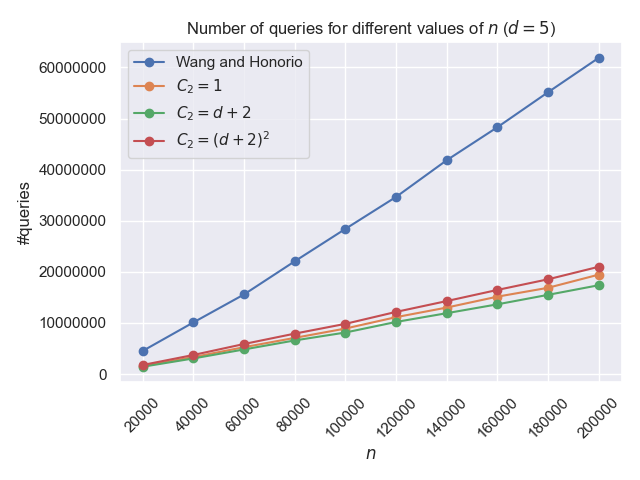}  
\end{subfigure}
\caption{Comparing Our Algorithm's number of rounds (left) and total queries (right) with Wang and Honorio's \cite{DBLP:conf/allerton/WangH19}, for fixed $d=5$ and varying $n$. }
\label{fig:fig_diff_n}
\end{figure}

\begin{figure}[t]
\centering
\begin{subfigure}{.49\textwidth}
  \centering
  \includegraphics[width=\linewidth,trim= 0 4mm 0 3mm,clip]{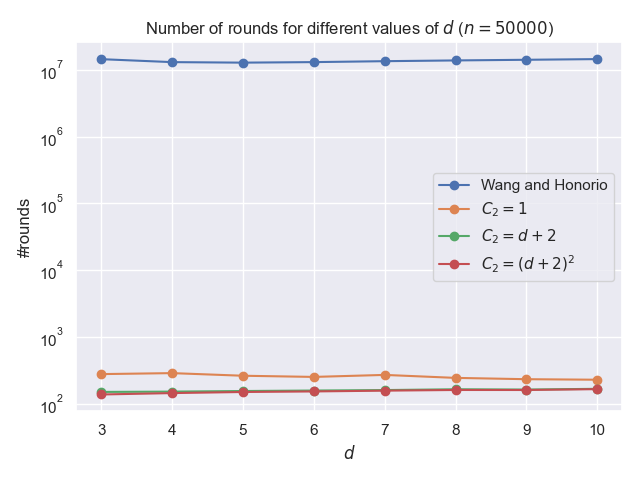}  
\end{subfigure}
\begin{subfigure}{.49\textwidth}
  \centering
  \includegraphics[width=\linewidth,trim= 0 4mm 0 3mm,clip]{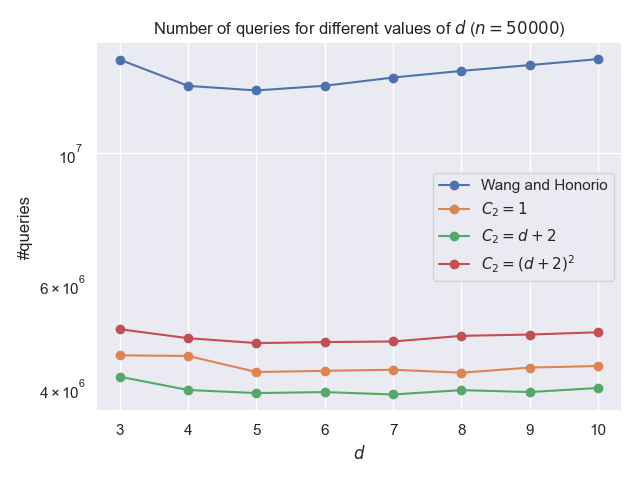}  
\end{subfigure}
\caption{Comparing our algorithm's number of rounds (left) and total queries (right) with Wang and Honorio's \cite{DBLP:conf/allerton/WangH19}, for $n=50000$ and varying values for $d$. } 
\label{fig:fig_diff_d}
\end{figure}

In \cref{fig:fig_diff_d}, we compare \cref{alg:rooted-reconstruct} with the one by Wang and
Honorio \cite{DBLP:conf/allerton/WangH19} for fixed size and varying values
of $d$. Again, this supports our theoretical
findings that our algorithm 
achieves both a significant parallel speedup and
a simultaneous improvement in the number of total queries. 

\begin{figure}[t]
  \centering
\begin{subfigure}{.49\textwidth}
  \centering
  \includegraphics[width=\linewidth,trim= 0 4mm 0 3mm,clip]{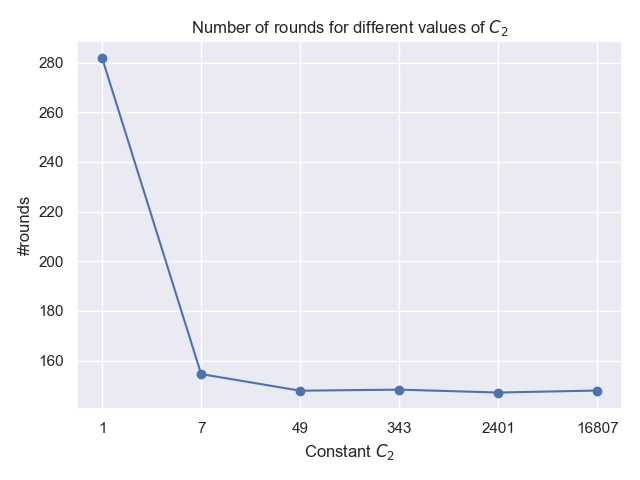}  
\end{subfigure}
\begin{subfigure}{.49\textwidth}
  \centering
  \includegraphics[width=\linewidth,trim= 0 4mm 0 3mm,clip]{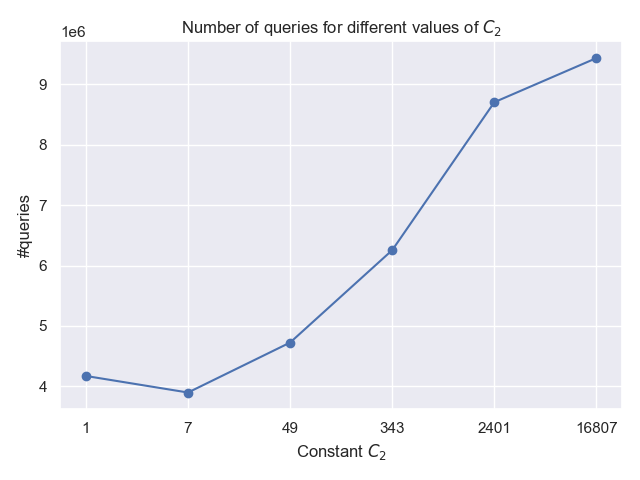}  
\end{subfigure}
\caption{Change in the total number of rounds (left) and total number of queries (right) when running our algorithm for varying values of $C_2$  ($n=50000$, $d=5$).} 
\label{fig:fig_diff_c}
\end{figure}

In \cref{fig:fig_diff_c}, we study the behavior of \cref{alg:rooted-reconstruct} under different values of $C_2$, so
as to experimentally find the
best value for $C_2$. 
While our high probability analysis requires $C_2 \approx (d+2)^4$,  \cref{fig:fig_diff_c} suggests that we do not need that 
high probability reassurance in practice, and we can use smaller sample  to
reduce the total number of queries.

\subparagraph*{Real Data.}~Our experiments on real-world biological phylogenetic trees also confirm the superiority of our algorithm in terms of performance as compared to the one by Wang and Honorio \cite{DBLP:conf/allerton/WangH19}. Similar to our experiments using relative-distance queries, we used a dataset comprised of trees from the phylogenetic library TreeBase~\cite{piel2009treebase}. \Cref{fig:phyloexp_path_queries} summarizes our experimental results, where each data point corresponds to an average performance of 3 runs on the same tree. Our algorithm is superior in both queries and rounds for all the values of $C_2$ we tried: $C_2 \in \{1, d+2, (d+2)^2\}$. The best performance corresponds to $C_2=d+2=5$, which is the one illustrated in \cref{fig:phyloexp_path_queries}.

\begin{figure}[bt]
\centering
\includegraphics[width=.7\linewidth,trim=0 4mm 0 3mm,clip]{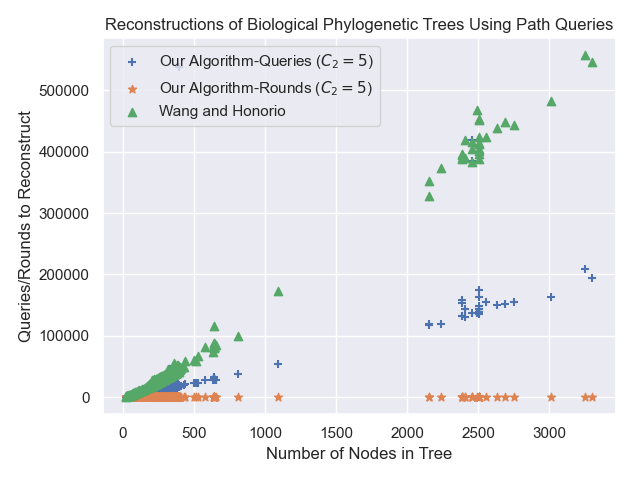}
\caption{A scatter plot comparing the number of queries and rounds of our algorithm and with the one by Wang and Honorio \cite{DBLP:conf/allerton/WangH19} for real-world trees from TreeBase~\cite{piel2009treebase}. Since our algorithm is parallel, we include round complexity to serve as a comparison for the sequential complexity.}
\label{fig:phyloexp_path_queries}
\end{figure}

\clearpage
\bibliographystyle{plainurl}
\bibliography{all-refs}

\clearpage

\appendix
\section{Probabilistic Analysis for Tree Reconstruction using Path Queries}

\label{appendix:chernoff}
Here, we give the Chernoff-bound analysis that we omitted in the proof of 
lemma~\ref{lem:expected_d_times}.

\begin{lemma}
There exists a constant $C_2 > 0$, as used $K = C_2 \log {|V|}$in line~\ref{split_3} of 
Algorithm~\ref{alg:splitting-edge-finding}, such that if we take a sample $X$ of size $K$ from $V$, the following probability bounds always hold:

\setcounter{equation}{0}
\begin{equation}\label{eq:1p}
  \begin{cases}
        
        \Pr\left(count(s,X) \geq \frac{K}{d+1}\right) \geq 1 - \frac{1}{|V|^2} & \mbox{if }count(s,V) \geq \frac{|V|}{d}, \\ 
        \Pr\left(count(s,X) \leq K  \frac{d}{d+1}\right) \geq 1 - \frac{1}{|V|^2} & \mbox{if } count(s,V) \leq |V|  \frac{d-1}{d},\\
               \end{cases}
 \end{equation}
 \begin{equation}\label{eq:2p}
  \begin{cases}
 
        \Pr\left(count(s,X) < \frac{K}{d+1}\right) \geq 1 - \frac{1}{|V|^2} & \mbox{if }count(s,V) < \frac{|V|}{d+2}, \\
        \Pr\left(count(s,X) > K  \frac{d}{d+1}\right) \geq 1 - \frac{1}{ |V|^2} & \mbox{if }count(s,V) > |V|  \frac{d+1}{d+2} \\
        
        \end{cases}
 \end{equation}

\end{lemma}

\begin{proof}

Recall that $\tilde{D}(s)=count(s,X) \frac{|V|}{K} $ is an estimation of $D(s)=count(s,V)$, the number of descendants of
vertex $s$ in algorithm~\ref{alg:splitting-edge-finding}. For simplicity, we use $D(s)$ and $\tilde{D}(s)$ in the proof. Let $Z$ be sum of independent binary random variables with
expected value $E[Z]$. Using a Chernoff bound, we know that $\Pr\left[\big|Z - E[Z]\big| \geq \epsilon  E[Z]\right] \leq 2  e^{\frac{-1}{3}  \epsilon^{2}  E[Z]}$: 

In this case, our random variable is $Z = \tilde{D}(s) \frac{K}{|V|}$ and $E[Z] = D(s) \frac{K}{|V|}$. By reformulating a Chernoff bound, we have
\begin{equation} \label{eq:3}
  \Pr \left[\big|Z - D(s)  \frac{K}{|V|}\big| \geq \epsilon  D(s)  \frac{K}{|V|}\right] \leq 2  e^{\frac{-1}{3}  \epsilon^{2}  D(s)  \frac{K}{|V|}}  
\end{equation}

Now, we find the value of $C_2$ used in line~\ref{split_3} of algorithm~\ref{alg:splitting-edge-finding} to compute $K$, the size of the sample. We do this for each of the 4 cases distinguished in equations~\ref{eq:1},\ref{eq:2}:

 \textbf{Case 1}:  We want to prove that if $D(s) \geq \frac{|V|}{d} $, 
 
 then  $\Pr\left[\tilde{D}(s) \geq \frac{|V|}{(d+1)}\right] \geq 1 - \frac{1}{|V|^2}$:

Suppose $D(s) \geq \frac{|V|}{d}$; we prove that $\Pr\left[\tilde{D}(s) < \frac{|V|}{d+1}\right] < \frac{1}{|V|^2}$. 

If we set $\epsilon = \frac{1}{d+1}$, we show that 
\begin{equation}\label{eq:4}
   \Pr\left[\tilde{D}(s) < \frac{|V|}{d+1}\right] \leq \Pr\left[\left|Z - D(s)  \frac{K}{|V|}\right| \geq \epsilon  D(s)  \frac{K}{|V|}\right] 
\end{equation}

In order to prove this, given the facts that $\epsilon = \frac{1}{d+1}$, and $D(s) \geq \frac{|V|}{d}$, we show that, for any $\tilde{D}(s)$ such that the inequality $\tilde{D}(s) < \frac{|V|}{d+1}$ holds, then the inequality $\left[\left|Z - D(s)  \frac{K}{|V|}\right| \geq \epsilon  D(s)  \frac{K}{|V|}\right]$ also holds.

\begin{align*}
& \frac{d}{d+1}   D(s)  \frac{K}{|V|} \geq  \frac{K}{d+1}> Z && \left(D(s) \ge \frac{|V|}{d}, Z < \frac{K}{d+1} \right) \\ 
\Longrightarrow & \left[\left(1 - \frac{1}{d+1}\right)  \left( D(s)  \frac{K}{|V|} \right) \geq Z  \right] \\
\Longrightarrow & \left[\left(D(s)  \frac{K}{|V|} - Z\right) \geq \frac{1}{d+1}  D(s)  \frac{K}{|V|}\right]\\
\Longrightarrow & \left[\left|Z - D(s)  \frac{K}{|V|}\right| \geq \epsilon  D(s)  \frac{K}{|V|}\right] && \left(\epsilon = \frac{1}{d+1} \right)
\end{align*}

Thus, inequality~\ref{eq:4} is true. Combining inequalities \ref{eq:3} and \ref{eq:4}, we have that:

$$
\Pr\left[\tilde{D}(s) < \frac{|V|}{d+1}\right] \le 2  e^{\frac{-1}{3}  \epsilon^{2}  D(s)  \frac{K}{|V|}}
$$

Now, we find the value of $C_2$, such that for $K = C_2   \log{|V|}$:

$$
  2  e^{\frac{-1}{3}  \epsilon^{2}  D(s)  \frac{K}{|V|}} < \frac{1}{|V|^2}
$$

Taking logarithm on both sides and given that $D(s) \geq \frac{|V|}{d}$  and $\epsilon = \frac{1}{d+1}$, we have that:

\begin{gather*}
   \frac{1}{3 }  \epsilon^{2}  D(s)  \frac{K}{|V|} \geq \frac{1}{3}  \left(\frac{1}{d+1}\right)^{2}  \frac{|V|}{|d|}  \frac{K}{|V|}
   = \frac{1}{3}  \left(\frac{1}{d+1}\right)^{2}  \frac{K}{d}  > 2 \ln{(2  |V|)} \\
  \Longleftrightarrow K > 6  d  (d+1)^2  \ln{(2  |V|)} 
  \\ \xLongleftrightarrow{K = C_2  \log{|V|}}
   C_2 > \frac{6  d  (d+1)^2  \ln{(2  |V|)}}{  \log{|V|}}  
\end{gather*}

Thus, $C_2$ is not more than a constant.

\noindent\textbf{Case 2:} We want to prove that if $D(s) \leq \frac{|V|  (d-1)}{d}$, 

then  $\Pr\left[\tilde{D}(s) \leq \frac{|V|  d}{(d+1)}\right] \geq 1 - \frac{1}{|V|^2}$:

Suppose $D(s) \leq \frac{|V|  (d-1)}{d}$; we prove that $\Pr\left[\tilde{D}(s)>\frac{|V|  d}{(d+1)}\right] < \frac{1}{|V^2}$.

Reminding that $Z= \tilde{D}(s)  \frac{K}{|V|}$, if we set $\epsilon = \frac{K}{d  (d+1)  E[Z]}$, we show:

\begin{equation}\label{eq:5}
   \Pr\left[\tilde{D}(s) > \frac{\left|V\right|  d}{\left(d+1\right)}\right] \leq \Pr\left[\left|Z - D(s)  \frac{K}{|V|}\right| \geq \epsilon  D(s)  \frac{K}{|V|}\right] 
\end{equation}

In order to prove this, given the facts that  $\epsilon = \frac{K}{d  (d+1)  E[Z]}$, and $D(s) \leq \frac{|V|  \left(d-1\right)}{d}$, we show that, for any $\tilde{D}(s)$ such that the inequality $Z = \tilde{D}(s)  \frac{K}{|V|} > \frac{K  d}{d+1}$ holds, then the inequality $\left[\left|Z - D(s)  \frac{K}{|V|}\right| \geq \epsilon  D(s)  \frac{K}{|V|}\right]$ also holds.

Given that $Z> \frac{K  d}{d+1}$ and that $D(s) \leq \frac{|V|  \left(d-1\right)}{d}$, we have:

\begin{gather*}
   \left|Z - D(s)  \frac{K}{|V|}\right| = \left(Z - D(s)  \frac{K}{|V|}\right) \\ > \frac{K  d}{d+1} - \frac{|V|  \left(d-1\right)}{d}  \frac{K}{|V|} = \frac{K}{d  \left( d+1\right)} 
\end{gather*}

Therefore, using the facts that $\epsilon = \frac{K}{d  \left( d+1\right)  E[Z]}$ and that $E[Z]=D(s)  \frac{K}{|V|}$, we can say:
$$  \left|Z - D(s)  \frac{K}{|V|}\right| \geq \epsilon  D(s)  \frac{K}{|V|} $$

Thus, inequality~\ref{eq:5} is true. Combining inequalities \ref{eq:3} and \ref{eq:5}, we have:

$$
\Pr\left[\tilde{D}(s) > \frac{|V|  d}{\left(d+1\right)}\right] \le 2  e^{\frac{-1}{3}  \epsilon^{2}  D(s)  \frac{K}{|V|}}
$$

Now, we find the value of $C_2$, such that for $K = C_2 \log{|V|}$:

$$
  2  e^{\frac{-1}{3}  \epsilon^{2}  D(s)  \frac{K}{|V|}} < \frac{1}{|V|^2}
$$

Taking logarithm on both sides and given that $\epsilon  E[Z]  = \frac{K}{d  (d+1)}$ and $E[Z] = D(s) \frac{K}{|V|} \leq \frac{K  (d-1)}{d}$, we can obtain $\epsilon \geq \frac{1}{(d-1)  (d+1)}$, therefore:

\begin{gather*}
   \frac{1}{3 }  \epsilon^{2}  D(s)  \frac{K}{|V|} \geq \frac{1}{3}  \frac{1}{(d-1)  (d+1)}  \frac{|K|}{d  (d+1)} > 2 \ln{(2  |V|)} \\ 
  \Longleftrightarrow   K > 6  d  (d-1)  (d+1)^2  \ln{(2  |V|)} \\
   \xLongleftrightarrow{K = C_2 \log{|V|}} 
   C_2 > \frac{6  d \dot (d-1)  (d+1)^2  \ln{(2  |V|)}}{  \log{|V|}}  
\end{gather*}

Thus, $C_2$ is not more than a constant.

\noindent\textbf{Case 3:} We want to show that if $D(s) < \frac{|V|}{d+2}$, 

then  $\Pr\left[\tilde{D}(s) < \frac{|V|}{(d+1)}\right] \geq 1 - \frac{1}{|V|^2}$:

Suppose $D(s) < \frac{|V|}{d+2}$; we prove that $\Pr\left[\tilde{D}(s) \geq \frac{|V|}{d+1}\right] < \frac{1}{|V|^2}$.

Reminding that $Z = \tilde{D}(s)  \frac{K}{|V|}$, if we set $\epsilon = \frac{K}{(d+1)  (d+2)  E[Z]}$, we show:

\begin{equation}\label{eq:6}
   \Pr\left[\tilde{D}(s) \geq \frac{\left|V\right| }{d+1}\right] \leq \Pr\left[\left|Z - D(s)  \frac{K}{|V|}\right| \geq \epsilon  D(s)  \frac{K}{|V|}\right] 
\end{equation}

In order to prove this, given the facts that $\epsilon = \frac{K}{(d+1)  (d+2)  E[Z]}$, and $D(s) < \frac{|V|}{d+2}$ , we show that, for any $\tilde{D}(s)$ such that the inequality $Z = \tilde{D}(s)  \frac{K}{|V|} \geq \frac{K}{d+1}$ holds, then the inequality $\left[\left|Z - D(s)  \frac{K}{|V|}\right| \geq \epsilon  D(s)  \frac{K}{|V|}\right]$ also holds.

Given that $Z \geq \frac{K}{d+1}$ and that $D(s) < \frac{|V|}{d+2}$, we can say:

\begin{gather*}
    \left|Z - D(s)  \frac{K}{|V|}\right| = \left(Z - D(s)  \frac{K}{|V|}\right) \\ > \frac{K}{d+1} - \frac{|V|}{d+2}  \frac{K}{|V|} = \frac{K}{\left(d+1\right)  \left( d+2\right)}
\end{gather*}

Therefore, using the facts that $\epsilon = \frac{K}{\left(d+1\right)  \left( d+2\right)  E[Z]}$ and that $E[Z]=D(s)  \frac{K}{|V|}$, we have:
$$  \left|Z - D(s)  \frac{K}{|V|}\right| \geq \epsilon  D(s)  \frac{K}{|V|} $$

Thus, inequality~\ref{eq:6} is true. Combining inequalities \ref{eq:3} and \ref{eq:6}, we can say:

$$
\Pr\left[\tilde{D}(s) \geq \frac{|V|}{d+1}\right] \le 2  e^{\frac{-1}{3}  \epsilon^{2}  D(s)  \frac{K}{|V|}}
$$

Now, we find the value of $C_2$, such that for $K = C_2  \log{|V|}$:

$$
  2  e^{\frac{-1}{3}  \epsilon^{2}  D(s)  \frac{K}{|V|}} < \frac{1}{|V|^2}
$$

Taking logarithm on both sides and given that 

$\epsilon  E[Z]= \frac{K}{(d+1)  (d+2)}$ and $E[Z] = D(s) \frac{K}{|V|} < \frac{K}{d+2}$, we can obtain $\epsilon > \frac{1}{(d+1)}$, therefore:

\begin{gather*}
   \frac{1}{3 }  \epsilon^{2}  D(s)  \frac{K}{|V|} \geq \frac{1}{3}  \frac{1}{(d+1)}  \frac{|K|}{(d+2)  (d+1)} > 2 \ln{(2  |V|)} \\ 
  \Longleftrightarrow  K > 6  (d+2)  (d+1)^2  \ln{(2  |V|)} \\
   \xLongleftrightarrow{K = C_2   \log{|V|}} 
   C_2 > \frac{6  (d+2)  (d+1)^2  \ln{(2  |V|)}}{  \log{|V|}}  
\end{gather*}

Thus, $C_2$ is not more than a constant.

\noindent\textbf{Case 4:} We want to prove that if $D(s) > \frac{|V|  (d+1)}{d+2}$, 

then  $\Pr\left[\tilde{D}(s) > \frac{|V|  d}{(d+1)}\right] \geq 1 - \frac{1}{|V|^2}$:

Suppose $D(s) > \frac{|V|  (d+1)}{d+2}$; we prove that $\Pr\left[\tilde{D}(s) \leq \frac{|V|  d}{(d+1)}\right] \geq 1 - \frac{1}{|V|^2}$.

Reminding that $Z= \tilde{D}(s)  \frac{K}{|V|}$, if we set $\epsilon = \frac{K}{(d+1)  (d+2)  E[Z]}$, we show:

\begin{equation}\label{eq:7}
   \Pr\left[\tilde{D}(s) \leq \frac{|V|  d}{(d+1)}\right] \leq \Pr\left[\left|X - D(s)  \frac{K}{|V|}\right| \geq \epsilon  D(s)  \frac{K}{|V|}\right] 
\end{equation}

In order to prove this, given the facts that $\epsilon = \frac{K}{(d+1)  (d+2)  E[Z]}$, and $D(s) > \frac{|V|  \left(d+1\right)}{d+2}$ , we show that, for any $\tilde{D}(s)$ such that the inequality $Z = \tilde{D}(s)  \frac{K}{|V|} \leq \frac{K  d}{d+1}$ holds, then the inequality $\left[\left|Z - D(s)  \frac{K}{|V|}\right| \geq \epsilon  D(s)  \frac{K}{|V|}\right]$ also holds.

Given that $Z \leq \frac{K  d}{d+1}$ and that $D(s) > \frac{|V|  (d+1)}{d+2}$, we can say:

\begin{gather*}
  \left|X - D(s)  \frac{K}{|V|}\right| = \left(D(s)  \frac{K}{|V|} - Z \right) \\ > \frac{|V|  (d+1)}{d+2}  \frac{K}{|V|} - \frac{K  d}{d+1}= \frac{K}{\left(d+1\right)  \left( d+2\right)}  
\end{gather*}

Therefore, using the facts that $\epsilon = \frac{K}{\left(d+1\right)  \left( d+2\right)  E[Z]}$ and that $E[Z]=D(s)  \frac{K}{|V|}$, we have:
$$  \left|Z - D(s)  \frac{K}{|V|}\right| \geq \epsilon  D(s)  \frac{K}{|V|} $$

Thus, inequality~\ref{eq:7} is true. Combining inequalities \ref{eq:3} and \ref{eq:7}, we can see:

$$
\Pr\left[\tilde{D}(s) \leq \frac{|V|  d}{(d+1)}\right] \le 2  e^{\frac{-1}{3}  \epsilon^{2}  D(s)  \frac{K}{|V|}}
$$

Now, we find the value of $C_2$, such that for $K = C_2  \log{|V|}$:

$$
  2  e^{\frac{-1}{3}  \epsilon^{2}  D(s)  \frac{K}{|V|}} < \frac{1}{|V|^2}
$$

Taking logarithm on both sides and given that 

$\epsilon  E[Z]  = \frac{K}{(d+1)  (d+2)}$ and $E[Z] = D(s) \frac{K}{|V|} \leq K$, we can obtain $\epsilon \geq \frac{1}{(d+1)  (d+2)}$, therefore:

\begin{gather*}
   \frac{1}{3 }  \epsilon^{2}  D(s)  \frac{K}{|V|} \geq \frac{1}{3}  \frac{1}{(d+1)  (d+2)}  \frac{|K|}{(d+2)  (d+1)} > 2 \ln{(2  |V|)} \\ 
  \Longleftrightarrow  K > 6  (d+2)^2  (d+1)^2  \ln{(2  |V|)} \\
   \xLongleftrightarrow{K = C_2 \log{|V|}} 
   C_2 > \frac{6  (d+2)^2  (d+1)^2  \ln{(2  |V|)}}{  \log{|V|}}  
\end{gather*}

Thus, $C_2$ is not more than a constant.

Therefore, it's enough to choose $C_2$ as the maximum of these 4 constants at the beginning of the algorithm.

\end{proof}

\end{document}